\documentclass[draftclsnofoot]{IEEEtran}
 \onecolumn
\usepackage{latexsym}
\usepackage{cite}
\usepackage{color}
\usepackage{bm}
\usepackage{amsmath, amssymb}
\usepackage[dvips]{graphics}
\usepackage{graphicx}
\usepackage{psfrag}
 \allowdisplaybreaks

\newtheorem{theorem}{Theorem}
\newtheorem{lemma}[theorem]{Lemma}

\newtheorem{corollary}[theorem]{Corollary}

\usepackage{slashbox}

\begin{document}

\newcommand{\be}{\begin{equation}}
\newcommand{\ee}{\end{equation}}
\newcommand{\bea}{\begin{eqnarray}}
\newcommand{\eea}{\end{eqnarray}}
\newcommand{\beaa}{\begin{eqnarray*}}
\newcommand{\eeaa}{\end{eqnarray*}}

\title{Capacity of a POST Channel with and without Feedback}
\author{Haim H. Permuter,  Himanshu Asnani and Tsachy Weissman\\
\thanks{H. Permuter is with the department of Electrical and Computer Engineering,
 Ben-Gurion University of the Negev, Beer-Sheva, Israel (haimp@bgu.ac.il).
H. Asnani and T. Weissman are with the  Department of Electrical
Engineering, Stanford University, CA, USA (asnani@stanford.edu,
tsachy@stanford.edu). This paper was presented in part at the 2013
IEEE International Symposium on Information Theory.}
}



%


\maketitle \vspace{-1.4cm}

\begin{abstract}
We consider finite state channels  where the state of the channel is
its previous output. We refer to these as POST (Previous Output is
the STate) channels. We first focus  on POST($\alpha$) channels.
These channels have binary inputs and outputs, where the state
determines if the channel behaves as a $Z$ or an $S$
channel, both with parameter $\alpha$. 
We show that the non feedback capacity of the POST($\alpha$) channel
equals its feedback capacity, despite the memory of the channel. The
proof of this surprising result is based on showing that the induced
output distribution, when maximizing the directed information in the
presence of feedback, can also be achieved by an input distribution
that does not utilize of the feedback. We show that this is a
sufficient condition for the feedback capacity to equal the non
feedback capacity for any finite state channel. We show that the
result carries over from the POST($\alpha$) channel to a binary POST
channel where the previous output determines whether the current channel will be binary with parameters $(a,b)$ or $(b,a)$.
Finally, we show that, in general, feedback may increase the
capacity of a POST channel.

\end{abstract}

\vspace{-0.0cm}
\section{Introduction}

The capacity of a memoryless channel is very well understood. There
are many simple memoryless channels for which we know the capacity
analytically. These include the binary symmetric channel, the
erasure channel, the additive Gaussian channel and the $Z$ Channel.
Furthermore, using convex optimization tools, such as the
Blahut-Arimoto algorithm \cite{Blahut72,Arimoto72}, we can
efficiently compute the capacity of any memoryless channel with a
finite alphabet. However, in the case of channels with memory, the
exact capacities are known for only a few channels, such as additive
Gaussian channels (water filling solution)
\cite{Shannon49_waterfilling,Pinsker60} and discrete additive
channels with memory \cite{Alajaji94}. In cases where feedback is
allowed, there are only a few more cases where the exact capacity is
known, such as the modulo-additive noise  channel, the additive
noise channel where the noise is a first-order autoregressive
moving-average Gaussian process
\cite{Kim10_Feedback_capacity_stationary_Gaussian}, the trapdoor
channel \cite{PermuterCuffVanRoyWeissman08}, and the Ising Channel
\cite{ElischoPermuter_Ising12}. If the state is known at the
decoder, then knowledge of the state at the encoder can be
considered as partial feedback, as considered and  solved in
\cite{Goldsmith_Varaiya97_state_fading} and in \cite{Chen05}.
\begin{figure}[h!]{
\psfrag{d1}[][][1]{$y_{i-1}=0$} \psfrag{d2}[][][1]{$y_{i-1}=1$}
\psfrag{x}[][][1]{$x_{i}$} \psfrag{y}[][][1]{$y_i$}
\psfrag{a1}[][][0.9]{$\alpha$} \psfrag{a2}[][][0.9]{$1-\alpha$}
  \centerline{ \includegraphics[width=8cm]{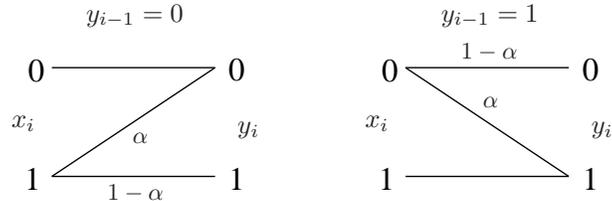}}
  \caption{ POST($\alpha$):
  If $y_{i-1}=0$ then the channel behaves as a $Z$ channel
  with parameter $\alpha$ and if $y_{i-1}=1$ then it behaves as an $S$ channel with parameter $\alpha$.}
 \label{Ychannel_alpha}
}\end{figure}

In this paper we introduce and consider a new family of channels
that we refer to as ``POST channels".
%
These are simple Finite State Channels (FSCs) where the state of the
channel is the previous output. In particular, we focus on a family
of POST channels that have binary inputs $\{X_i\}_{i\geq 1}$ and
binary outputs $\{Y_i\}_{i\geq 1}$ related as follows:

\begin{equation}
\text{if } X_i=Y_{i-1}, \text{ then } Y_i=X_i, \text{ else } Y_i=X_i
\oplus Z_i, \text{  where } Z_i\sim
\text{Bernnouli}\left(\alpha\right).
\end{equation}
We call these channels POST($\alpha$) and their behavior is depicted
in Fig. \ref{Ychannel_alpha}. When $y_{i-1}=0$, the current channel
behaves as a $Z$ channel with parameter $\alpha$ and when
$y_{i-1}=1,$ it behaves as an $S$ channel with parameter $\alpha$.
We refer to POST$(\frac{1}{2})$ as the {\it simple} POST channel.

The simple POST channel is similar to the Ising channel introduced by Berger and Bonomi
\cite{Berger90IsingChannel}, but rather than the previous input being the state of the channel,
here the state of the channel is the previous output.  This channel arose in the investigation of
controlled feedback in the setting of ``to feed or not to feed back"
\cite{AsnaniPermuterWeissman10_to_feed_or_not}. The POST channel can also be useful in modeling
memory affected by past channel outputs, as is the case in flash memory and other storage
devices.

In order to gain intuition for investigating the influence of
feedback on the simple POST channel, let us first consider a channel
with binary i.i.d. states $S_i$, distributed
Bernoulli($\frac{1}{2}$), where the channel behaves similarly to the
simple POST channel. When $S_{i-1}=0$, then the current channel
behaves as a $Z$ channel and when $S_{i-1}=1$, it behaves as an $S$
channel, as shown in Fig. \ref{f_similar_channel}. Similar to the
POST channel, we assume that the state of the channel is known to
the decoder; hence the output of the channel is $(Y_i,S_{i})$ (or,
equivalently from a capacity standpoint, $(Y_i,S_{i-1})$).
\begin{figure}[h!]{
\psfrag{d1}[][][1]{$s_{i-1}=0$} \psfrag{d2}[][][1]{$s_{i-1}=1$}
\psfrag{x}[][][1]{$x_{i}$} \psfrag{y}[][][1]{$y_i$}
\psfrag{a1}[][][0.9]{$\frac{1}{2}$}
\psfrag{a2}[][][0.9]{$\frac{1}{2}$}
\psfrag{TX}[][][1]{En}\psfrag{Xi}[][][0.9]{$X_i$}\psfrag{Yi}[][][0.9]{$Y_i,S_i\
\ $}\psfrag{P}[][][1]{$\ P(s_i,y_i|x_i)$}
\psfrag{Ch}[][][0.9]{Channel} \psfrag{Ps}[][][0.9]{$S_i$ i.i.d.
$\sim $Bernoulli($\frac{1}{2}$)} \psfrag{RX}[][][1]{De}
\psfrag{M}[][][0.9]{$M$} \psfrag{M2}[][][0.9]{$\hat M$}
\psfrag{0}[][][1]{0}\psfrag{1}[][][1]{1}
  \centerline{ \includegraphics[width=8cm]{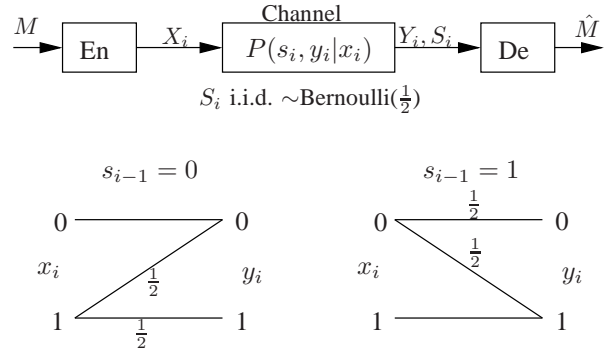}}
  \caption{ A channel similar  to the simple POST channel, except that the channel state, $\{S_i\}_{i\geq 1}$ is  i.i.d.
  Bernoulli($\frac{1}{2}$) independent of the input.}
 \label{f_similar_channel}
}\end{figure}

The non feedback capacity of this channel is simply
$C=\max_{P(x)}I(X;Y,S)$ and, because of symmetry, the input that
achieves the maximum is Bernoulli($\frac{1}{2}$) , resulting in a
capacity of  $H_b(\frac{1}{4})-\frac{1}{2}=0.3111$, where $H_b(p)$
is the binary entropy function. However, if there is perfect
feedback of $(Y_i,S_i)$ to the encoder, then the state of the
channel is known to the encoder and the capacity is simply the
capacity of the $Z$ (or $S$) channel, which is
$H_b(\frac{1}{5})-\frac{2}{5}=-\log_2 0.8=0.3219$. Evidently,
feedback increases the capacity of this channel.

The similarity between the channels may seem to hint that feedback
increases the capacity of the POST($\alpha$) channel as well.
Indeed, our initial interest in this channel was due to this belief,
in our quest for a channel with memory that would be amenable to
analysis under the ``to feed or not to feed'' framework of
constrained feedback in
\cite{AsnaniPermuterWeissman10_to_feed_or_not}, while exhibiting
non-trivial dependence of its capacity on the extent to which the
feedback is constrained.
However, numerical results based on the computational algorithm devised in
\cite{AsnaniPermuterWeissman10_to_feed_or_not} suggested that feedback does not increase the
capacity of the simple POST channel. This paper stemmed from our attempts to make sense of
these observations.

In order to prove that feedback does not increase the capacity of
some families of POST channels, we look at two convex optimization
problems: maximizing the directed information over regular input
distributions (non feedback case), i.e., $P(x^n)$ and, secondly,
over causal conditioning that is influenced by the feedback i.e.,
$P(x^n||y^{n-1})$. We show that a necessary and sufficient condition
for the solutions of the two optimization problems to achieve the
same value is that the induced output distributions $P(y^n)$ by the
respective optimal values $P^*(x^n)$ and $P^*(x^n||y^{n-1})$ are the
same. This necessary and sufficient condition that we establish, in
the generality of any finite state channel, follows from the KKT
conditions \cite[Ch. 5]{BoydOptimizationBook04} for convex
optimization problems.

%

The remainder of the paper is organized as follows. In Section
\ref{s_dircted_inco_causal_con}, we briefly present the definitions
of directed information and causal conditioning pmfs that we use
throughout the paper. In Section
\ref{s_max_directed_transorm_into_convex}, we show that the
optimization problem of maximizing the directed information over
causal conditioning pmfs is  convex. Additionally, using the KKT
conditions, we show that if the output distribution induced by the
conditional pmfs that achieve the maximum in the presence of
feedback  can also be induced by an input distribution that does not
use feedback, then feedback does not increase the capacity. In
Section \ref{s_capacity_y_channel} we compute the feedback capacity
of the POST($\alpha$) channel. Then we apply the result of Section
\ref{s_max_directed_transorm_into_convex} to show that it equals the
non feedback capacity  by establishing the existence of an input
distribution without feedback that induces the same output
distribution as the capacity achieving one for the feedback case. In
Section \ref{s_post_ab} we consider a binary POST$(a,b)$ channel
with two states; in each state there is a binary channel and the
channels have opposite parameters. The binary POST$(a,b)$ channel
generalizes the POST($\alpha$) channel and we show that  feedback
does not increase capacity  for this considerably larger class of
channels. In Section \ref{s_POST_increase_capacity}, we show that
unlike the the POST($a,b$), feedback may increase the capacity of
POST channels in general. In Section \ref{s_conclusion}, we conclude
and suggest some directions for further research on the family of
POST channels.

\section{Directed information, causal conditioning and notations\label{s_dircted_inco_causal_con}}


Throughout this paper, we denote random variables by capital letters
such as $X$. The probability $\Pr\{X=x\}$ is denoted by $p(x)$. We
denote the whole vector of probabilities by capital $P$, i.e.,
$P(x)$ is the probability vector of the random variable $X$. 

We use the {\it causal conditioning} notation $(\cdot||\cdot)$
developed by Kramer~\cite{Kramer98}. We denote by $p(x^n||y^{n-d})$
the probability mass function of $X^n = (X_1,\ldots, X_n)$,
\emph{causally conditioned} on $Y^{n-d}$ for some integer $d\geq 0$,
which is defined as
\begin{equation} \label{e_causal_cond_def}
p(x^n||y^{n-d}) := \prod_{i=1}^{n} p(x_i|x^{i-1},y^{i-d}).
\end{equation}
By convention, if $i < d$, then $y^{i-d}$ is set to null, i.e.,  if $i < d$ then
$p(x_i|x^{i-1},y^{i-d})$ is just $p(x_i|x^{i-1})$.  In particular, we use extensively the cases
$d=0,1$:
\begin{align}
p(x^n||y^{n}) &:= \prod_{i=1}^{n}p(x_i|x^{i-1},y^{i}),\\
p(x^n||y^{n-1}) &:= \prod_{i=1}^{n} p(x_i|x^{i-1},y^{i-1}).
\end{align}

The directed information was defined by Massey~\cite{Massey90}, inspired by Marko's work
\cite{Marko73} on bidirectional communication,  as
\begin{equation}\label{e_def_directed}
I(X^n\to Y^n) := \sum_{i=1}^n I(X^i;Y_i|Y^{i-1}).
\end{equation}
The directed information can also be rewritten as
\begin{equation}\label{e_kim_identity}
I(X^n \to Y^n) = \sum_{x^n,y^n}p(x^n||y^{n-1})p(y^n||x^n)\log
\frac{p(y^n||x^n)}{\sum_{x^n}p(x^n||y^{n-1})p(y^n||x^n)}.
\end{equation}
This is due to the definition of causal conditioning and the chain rule
\begin{equation}
p(x^n,y^n)=p(x^n||y^{n-1})p(y^n||x^n).
\end{equation}
We will make use the fact that directed information $I(X^n\to Y^n)$
is concave in $P(x^n||y^{n-1})$ for a fixed $P(y^n||x^n),$ which is
proved in Lemma \ref{l_concavity_dir} in Appendix
\ref{app_concavity}.

Directed information characterizes the capacity of point-to-point
channels with feedback~\cite{Chen05,Kim08_feedback_directed,
  TatikondaMitter_IT09, PermuterWeissmanGoldsmith09}. For channels where the state is a function of the output, of which the POST
  channel is a special case, it was shown
\cite{Chen05,PermuterCuffVanRoyWeissman08} that the feedback
capacity is given by
\begin{equation}\label{e_max_feedback}
C_{fb}=\lim_{n\to \infty} \frac{1}{n} \max_{P(x^n||y^{n-1})} I(X^n\to Y^n).
\end{equation}
On the other hand, without feedback the capacity is given by
\begin{equation}\label{e_max_feedback}
C=\lim_{n\to \infty} \frac{1}{n} \max_{P(x^n)} I(X^n\to Y^n),
\end{equation}
since the channel is indecomposable  \cite{Gallager68}. In the case where there is no feedback,
namely,  the Markov form $X_i-X^{i-1}-Y^{i-1}$ holds, $I(X^n\to Y^n)=I(X^n;Y^n),$ as shown in
\cite{Massey90}.


\section{Maximization of the directed information as a convex optimization problem
\label{s_max_directed_transorm_into_convex}} In order to show that
feedback does not increase the capacity of  POST channels, we
consider the two optimization problems:
\begin{equation}\label{e_opt_dir_feedback}
\max_{P(x^n||y^{n-1})} I(X^n\to Y^n)
\end{equation}
and
\begin{equation}\label{e_no_feedback}
\max_{P(x^n)} I(X^n\to Y^n).
\end{equation}
In this section, we  show that both problems are convex optimization
problems, and use the KKT condition to state a necessary and
sufficient condition for the two optimization problems to obtain the
same value.

A convex optimization problem, as defined in \cite[Ch.
4]{BoydOptimizationBook04}, is a problem of the form
\begin{align}\label{e_convex_fun}
&\mathrm{minimize \ \ \ \ }  f_0(x) \\ \notag &\mathrm{subject \ to } \ \  \ f_i(x) \leq b_i
\quad i=1,\cdots,k \\ \notag &\ \ \ \ \ \ \ \ \ \ \ \ \ \ \ g_j(x) = 0  \quad j=1,\cdots,l
\end{align}
where $f_0(x)$ and $\{f_i(x)\}_{i=1}^k$ are convex functions, and $\{g_j(x)\}_{j=1}^l$ are
affine.


In order to convert the optimization problem in
(\ref{e_opt_dir_feedback}) into a convex optimization problem, as
presented in (\ref{e_convex_fun}), we need to show that the set of
conditional pmfs $P(x^n||y^{n-1})$ can be expressed using
inequalities that contains only convex functions and equalities that
contains affine functions.

\begin{lemma}[Causal conditioning is a polyhedron] \label{l_cusal_polyhedron} The set of all causal conditioning distributions of the
form $P(x^n||y^{n-1})$ is a polyhedron in $\mathbb{R}^{|\mathcal
X|^n|\mathcal Y|^{n-1}}$ and is given by a set of linear equalities
and inequalities of the form:

\begin{equation}\label{e_causal1}
\begin{array}{ll}
p(x^n||y^{n-1})\geq 0,& \forall x^n, y^{n-1}, \\
 \sum_{x_{i+1}^n}p(x^n||y^{n-1})= \gamma_{x^{i},y^{i-1}},& \forall x^i,
y_{}^{n-1},
i\geq 1,\\
 \sum_{x_{1}^n}p(x^n||y^{n-1})= 1,& \forall  y_{}^{n-1}.
\end{array}
\end{equation}
\end{lemma}
Note that the two equalities in (\ref{e_causal1}) may be unified
into one if we add $i=0$ to the equality cases and we restrict the
corresponding $\gamma$ to be unity. Furthermore, for $n=1$ we obtain
the regular vector probability, i.e., $p(x)\geq0, \ \forall x$ and
$\sum_x P(x)=1$.
\begin{proof}
It is straightforward to see that every causal conditioning $P(x^n||y^n)$ satisfies these
equalities and inequalities. Now, we need to show that if an element in $\mathbb{R}^{|\mathcal
X|^n|\mathcal Y|^{n-1}}$, call it  $P(x^n||y^{n-1})$, satisfies (\ref{e_causal1}), then
$P(x^n||y^{n-1})$ is a causal conditioning pmf, namely, there exists a sequence of regular
conditioning  $\{P(x_i|x^{i-1},y^{i-1})\}_{i=1}^n$ such that
$p(x^n||y^{n-1})=\prod_{i=1}^{n}p(x_i|x^{i-1},y^{i-1}).$

Let  us define for all $x^{i},y^{i-1}$
\begin{align}
p(x_i|x^{i-1},y^{i-1})&\triangleq \frac{\sum_{x^{i+1}}^n p(x^n||y^n)}{\sum_{x^{i}}^np(x^n||y^n)}\nonumber \\
&=\frac{\gamma_{x^{i},y^{i-1}}}{\gamma_{x^{i-1},y^{i-2}}}.\label{e_p_cond}
\end{align}
Now, note that the vector probability  $P(x_i|x^{i-1},y^{i-1})$
defined in (\ref{e_p_cond}) satisfies $p(x_i|x^{i-1},y^{i-1})\geq0$,
and $\sum_{x_i}p(x_i|x^{i-1},y^{i-1})=1.$ Furthermore, observe that
\begin{align}
\prod_{i=1}^{n}p(x_i|x^{i-1},y^{i-1})=\prod_{i=1}^{n}\frac{\gamma_{x^{i},y^{i-1}}}{\gamma_{x^{i-1},y^{i-2}}}\nonumber
=\frac{\gamma_{x^{n},y^{n-1}}}{1}=\gamma_{x^{n},y^{n-1}} =p(x^n||y^{n-1}).
\end{align}

\end{proof}

Note that the optimization problem given in
(\ref{e_opt_dir_feedback}) is a convex optimization one since the
set of causal conditioning pmfs is a polyhedron (Lemma
\ref{l_cusal_polyhedron}) and the directed information is concave in
$P(x^n||y^{n-1})$ for a fixed $P(y^n||x^{n})$ \cite[Lemma
2]{NaissPermuter12_BA_DI}. Therefore, the KKT conditions \cite[Ch
5.5.3]{BoydOptimizationBook04} are necessary and sufficient. The
next theorem states these conditions explicitly for our setting.

\begin{theorem}[Necessary and sufficient conditions for maximizing the the directed information]
A set of necessary and sufficient conditions for an input
probability  $P(x^n||y^{n-1})$ to achieve the maximum  in
(\ref{e_max_feedback}) is that for some numbers $\beta_{y^{n-1}}$
\begin{align}
\sum_{y_n}p(y^n||x^n)\log \frac{p(y^n||x^n)}{ep(y^n)}=\beta_{y^{n-1}}, \ \  \forall x^n, y^{n-1},
\text{
if } p(x^n||y^{n-1})>0, \nonumber \\
\sum_{y_n}p(y^n||x^n)\log \frac{p(y^n||x^n)}{ep(y^n)}\leq\beta_{y^{n-1}}, \ \  \forall x^n,
y^{n-1}, \text{ if } p(x^n||y^{n-1})=0,\label{e_th_cond}
\end{align}
where $p(y^n)=\sum_{x^n}p(y^n||x^n)p(x^n||y^{n-1})$. Furthermore,
the maximum is given by
 \begin{equation}\label{e_th_cap}
\max_{P(x^n||y^{n-1})} I(X^n\to Y^n)=\sum_{y^{n-1}}\beta_{y^{n-1}}+1.
\end{equation}

\end{theorem}

For $n=1$ we obtain a known result proved by Gallager \cite[Theorem
4.5.1]{Gallager68} that states that a sufficient and necessary
condition for $P^*(x)$ to achieve $\max_{P(x)}I(X;Y)$ is that
\begin{equation}
\sum_{y}p(y|x)\log \frac{p(y|x)}{p(y)}=C, \ \  \forall x \text{ if  } p^*(x)>0,
\end{equation}
and
\begin{equation}
\sum_{y}p(y|x)\log \frac{p(y|x)}{p(y)}\leq C, \ \  \forall x \text{ if } p(x)=0,
\end{equation}
for some $C$. Furthermore, $C=\max_{P(x)}I(X;Y).$

\begin{proof}
Using the fact that a causal conditioning pmf is a polyhedron (Lemma \ref{l_cusal_polyhedron}),
we can write the maximization of the directed information as a standard convex optimization
problem:

\begin{equation}\label{e_opt}
\begin{array}{ll}
\text{minimize} & -\sum_{x^n,y^n}p(x^n||y^{n-1})p(y^n||x^n)\log \frac{p(y^n||x^n)}{\sum_{x^n}p(x^n||y^{n-1})p(y^n||x^n)}  \\
\text{s.t. } & -p(x^n||y^{n-1})\leq 0,\  \forall x^n, y^{n-1},\\
&\sum_{x_{i+1}^n}p(x^n||y^{n-1})- \gamma_{x^{i},y^{i-1}}=0,\ \forall x^i, y_{}^{n-1}, i\geq 1\\
&\sum_{x_{}^n}p(x^n||y^{n-1})= 1,\ \forall  y_{}^{n-1}.
\end{array}
\end{equation}

The Lagrangian is defined as
\begin{align}
L=&-\sum_{x^n,y^n}p(x^n||y^{n-1})p(y^n||x^n)\log
\frac{p(y^n||x^n)}{\sum_{x^n}p(x^n||y^{n-1})p(y^n||x^n)}-
\sum_{x^n,y^{n-1}}\lambda_{x^n,y^{n-1}}p(x^n||y^{n-1}) \nonumber \\
& +\sum_{i=0}^n \sum_{x^i, y^{n-1}}\nu_{x^i, y_{}^{n-1}}\left( \sum_{x_{i+1}^n}p(x^n||y^{n-1})-
\gamma_{x^{i},y^{i-1}}\right)+ \sum_{y^{n-1}}\nu_{y_{}^{n-1}}\left( \sum_{x_{}^n}p(x^n||y^{n-1})-
1\right),
\end{align}
where $\lambda_{x^n,y^{n-1}}\geq0.$ The KKT conditions for this problem are
\begin{align}\label{e_kkt}
\frac{\partial L}{\partial L p(x^n||y^{n-1})}& =0, \ \forall x^n, y_{}^{n-1},\nonumber \\
\frac{\partial L}{\partial L \gamma_{x^{i},y^{i-1}}}& =0, \ \forall x^n, y_{}^{n-1},\nonumber \\
\lambda_{x^n,y^{n-1}}& \geq0,  \ \forall x^n, y_{}^{n-1}, \nonumber \\
\lambda_{x^n,y^{n-1}}p(x^n||y^{n-1})& =0,  \ \forall x^n, y_{}^{n-1}
\end{align}
and that $P(x^n||y^{n-1})$ is a valid causal conditioning pmf, namely, satisfies the constraint
of (\ref{e_opt}). Now, we need to show that the KKT conditions described above are equivalent to
(\ref{e_th_cond}). Let us compute the derivatives in order to write the conditions in
(\ref{e_kkt}) more explicitly:
\begin{align}
\frac{\partial L}{\partial
p(x^n||y^{n-1})}=&-\sum_{y_n}p(y^n||x^n)\log{P(y^n||x^n)}-\sum_{y_n}p(y^n||x^n)\log
\frac{1}{\sum_{x^n}p(x^n||y^{n-1})P(y^n||x^n)} \nonumber \\
& +\sum_{x'^n,y_n}p(x'^n||y^{n-1})P(y^{n}||x'^n) \frac{1}{\sum_{\tilde x^n}p(\tilde
x^n||y^{n-1})p(y^n||\tilde x^n)}p(y^n||x^n)\nonumber \\
&- \lambda_{x^n,y^{n-1}}+ \sum_{i=0}^n \nu_{x^i, y_{}^{n-1}} + \nu_{ y_{}^{n-1}} \nonumber \\
 =& -\sum_{y_n}p(y^n||x^n)\log\frac{{p(y^n||x^n)}}{\sum_{x^n}p(x^n||y^{n-1})p(y^n||x^n)}+\sum_{y_n}p(y^n||x^n)\nonumber \\
&- \lambda_{x^n,y^{n-1}}+ \sum_{i=0}^n \nu_{x^i, y_{}^{n-1}} + \nu_{ y_{}^{n-1}} \nonumber \\
 =& -\sum_{y_n}p(y^n||x^n)\log\frac{{p(y^n||x^n)}}{ep(y^n)}- \lambda_{x^n,y^{n-1}}+ \sum_{i=0}^n \nu_{x^i, y_{}^{n-1}} + \nu_{ y_{}^{n-1}} 
\end{align}
\begin{equation}
\frac{\partial L}{\partial \gamma_{x^{i},y^{i-1}}}=-\sum_{i=0}^n
\nu_{x^i, y_{}^{n-1}}.
\end{equation}
Hence, the KKT conditions given in (\ref{e_kkt}) become
\begin{align}\label{e_kkt1}
\sum_{y_n}p(y^n||x^n)\log\frac{{p(y^n||x^n)}}{ep(y^n)}=&\nu_{ y_{}^{n-1}} \ \forall x^n, y_{}^{n-1} \text{ if } P(x^n||y^{n-1})>0, \nonumber \\
\sum_{y_n}p(y^n||x^n)\log\frac{{p(y^n||x^n)}}{ep(y^n)}\leq &\nu_{ y_{}^{n-1}} \ \forall x^n,
y_{}^{n-1} \text{ if } p(x^n||y^{n-1})=0,
\end{align}
which is exactly (\ref{e_th_cond}). Now, to obtain (\ref{e_th_cap})
we use (\ref{e_kkt1}) and observe that
\begin{align}\label{e_kkt2}
\sum_{x^n}\sum_{y^{n-1}} p(x^n||y^{n-1})
\sum_{y_n}p(y^n||x^n)\log\frac{{p(y^n||x^n)}}{ep(y^n)}=&\sum_{X^n}\sum_{y^{n-1}}
p(x^n||y^{n-1})\nu_{ y_{}^{n-1}}.
\end{align}
Indeed, the LHS of (\ref{e_kkt2}) equals to $I(X^n\to Y^n)-1$. The
RHS equals to $\sum_{y^{n-1}} \nu_{ y_{}^{n-1}},$ since
$\sum_{x^n}p(x^n||y^{n-1})=1$ for all $y^{n-1}$ and, therefore,
(\ref{e_kkt2}) implies that (\ref{e_th_cap}) holds.
\end{proof}

The next corollary is the main tool we use in this paper to prove
that the feedback capacity and the non feedback capacity of a
channel are equal.
\begin{corollary}\label{cor_pyn}
Let $P^*(x^n||y^{n-1})$ be a pmf that all its elements are positive
and that achieves the maximum of $ \max_{P(x^n||y^{n-1})} I(X^n\to
Y^n),$ and let $P^*(y^n)$ be the output probability induced by
$P^*(x^n||y^{n-1}).$ If for any $n$ there exists an input
probability distribution $P(x^n)$ such that
 \begin{equation}
p^*(y^n)=\sum_{x^n}p(y^n||x^n)p(x^n),
\end{equation}
 then the feedback capacity and the nonfeedback capacity
are the same.
\end{corollary}
\begin{proof}
Note that the sufficient and necessary condition given in (\ref{e_th_cond}) depends only on the
channel causal conditioning pmf $P(y^n||x^n)$ and the output pmf $P(y^n)$. Furthermore, note that
if  (\ref{e_th_cond}) is satisfied then
\begin{align}
\sum_{y^n}p(y^n||x^n)\log \frac{p(y^n||x^n)}{ep(y^n)}=\sum_{y^{n-1}}\beta_{y^{n-1}}, \ \  \forall
x^n,
\end{align}
and since for the non-feedback case $p(y^n||x^n)=p(y^n|x^n)$, $\forall\ (x^n,y^n)$, we obtain
\begin{align}
\sum_{y^n}p(y^n|x^n)\log \frac{p(y^n|x^n)}{ep(y^n)}=\sum_{y^{n-1}}\beta_{y^{n-1}}, \ \  \forall
x^n.
\end{align}
This means that the KKT conditions of $\max_{P(x^n)}I(X^n;Y^n)$ are satisfied. Furthermore, the
 maximum value for both optimization problems is $\sum_{y^{n-1}}\beta_{y^{n-1}}+1$ and,
 therefore, they are equal.
\end{proof}

\section{Capacity of the POST($\alpha$) channel with and without feedback \label{s_capacity_y_channel}}

\begin{lemma}[Feedback capacity]\label{l_posta_feedback}
The feedback capacity of the POST($\alpha$) channel  is the same as
of the memoryless $Z$ channel with parameter $\alpha$, which is
$C=-\log_2 c$ where
\begin{equation}\label{e_normalization_coeeficient}
c=(1+\bar{\alpha} \alpha^{\frac{\alpha}{\bar{\alpha}}})^{-1}.
\end{equation}
The behavior  of the capacity as a function of $\alpha$ is depicted
in Fig. \ref{f_post_alpha}.
\end{lemma}
\begin{figure}[h!]{
\psfrag{a}[][][0.8]{}
\psfrag{c}[][][0.9]{Capacity of POST($\alpha$)}
\psfrag{alpha}[][][1]{$\alpha$}
  \centerline{
  \includegraphics[width=6cm]{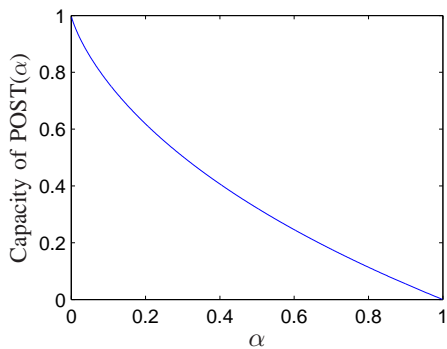}}
  \caption{The capacity of the POST($\alpha$) channel with and without feedback. This is also the capacity of
  the $Z$ channel with parameter $\alpha$\label{f_post_alpha}}
}\end{figure}
 {\it Proof of the achievability and the capacity of
$Z$ channel with parameter $\alpha$:} the achievability proof is
trivial since both the encoder and the decoder know if the channel
behaves as a $Z$ or an $S$ channel. The input probability that
maximizes the mutual information, i.e. $\arg \max_{P_X}I(X;Y),$ for
the memoryless $Z$ channels with parameter $\alpha$ is
\begin{equation}
P(x=1)=c{\alpha^{\frac{\alpha}{\bar{\alpha}}}} \ \ \
P(x=0)=c({1-\alpha^{\frac{1}{\bar{\alpha}}}}),
\end{equation}
where $c$ is a normalization coefficient and is given in
(\ref{e_normalization_coeeficient}). The output probability for the
$Z$ channel with parameter $\alpha$ is
\begin{align}\label{e_py}
&P(y=1)=\bar \alpha P(x=1)=c{\bar \alpha
\alpha^{\frac{\alpha}{\bar{\alpha}}}} \nonumber \\
&P(y=0)=1-c{\bar \alpha \alpha^{\frac{\alpha}{\bar{\alpha}}}}=c.
\end{align}
The capacity of the $Z$ channel with parameter $\alpha$ is
\begin{align}
C=\max_{P_X}I(X;Y)=-\log_2 c,
\end{align}
which is also an achievable rate for the POST($\alpha$) channel with
feedback. \hfill \QED

{\it Proof of converse:}
 The upper bound is given in the following set of
equalities and inequalities:

\begin{eqnarray}
C_{fb}&\stackrel{(a)}{=}&\lim_{n\to \infty} \max_{P(x^n||y^{n-1})} \frac{1}{n} I(X^n\to Y^n)\nonumber\\
&\stackrel{(b)}{=}&\lim_{n\to \infty} \max_{P(x^n||y^{n-1})}  \frac{1}{n} I(X^n\to Y^n|Y_0)\nonumber\\
&=&\lim_{n\to \infty} \max_{P(x^n||y^{n-1})}  \frac{1}{n} \sum_{i=1}^n H(Y_i|Y^{i-1})-H(Y_i|Y^{i-1},X^{i}) \nonumber\\
&\stackrel{(c)}{\leq} &\lim_{n\to \infty} \max_{P(x^n||y^{n-1})} \frac{1}{n} \sum_{i=1}^n H(Y_i|Y_{i-1})-H(Y_i|Y_{i-1},X_{i}) \nonumber\\
&\stackrel{}{=} &\lim_{n\to \infty} \max_{P(x^n||y^{n-1})} \sum_{i=1}^n P(y_i=0)I(X_i;Y_i|y_{i-1}=0)+P(y_i=1)I(X_i;Y_i|y_{i-1}=1) \nonumber\\
&\stackrel{(d)}{=} &\lim_{n\to \infty} \max_{\{P(x_i|y_{i-1})\}_{i\geq1}} \sum_{i=1}^n P(y_i=0)I(X_i;Y_i|y_{i-1}=0)+P(y_i=1)I(X_i;Y_i|y_{i-1}=1) \nonumber\\
&\stackrel{}{=}&\max_{P(x_i|y_{i-1}=0)} P(y_i=0)I(X_i;Y_i|y_{i-1}=0)+\max_{P(x_i|y_{i-1}=1)} P(y_i=1)I(X_i;Y_i|y_{i-1}=1), \text{ for some } i \nonumber\\
&\stackrel{(e)}{=}& -P(y_i=0)\log_2 c-P(y_i=1)\log_2 c,\nonumber\\
&\stackrel{}{=}& -\log_2 c,
\end{eqnarray}
where (a) follows from the capacity formula given in
\cite{PermuterWeissmanGoldsmith09,Chen05},  (b) from the inequality
$|I(X^n\to Y^n)-I(X^n\to Y^n|S)|\leq H(S)$ \cite[Lemma
4]{PermuterWeissmanGoldsmith09},  (c) from the fact that
conditioning reduces entropy and from the Markov chain
$Y_i-(X_i,Y_{i-1})-(X^{i-1},Y^{i-2})$,  (d) from the fact that the
set of causal conditioning $P(x^n||y^{n-1})$ is equivalent to the
set of $\{P(x_i|x^{i-1},y^{i-1})\}$ and in this particular case it's
enough to maximize only over $\{P(x_i|y_{i-1})\}$ because of the
objective. Finally, Step (e) follows from the capacity of the
memoryless $Z$-channel with parameter $\alpha$. \hfill \QED

Note that the induced $Y_i$ is a Markov chain with transition
probability $c{\bar \alpha \alpha^{\frac{\alpha}{\bar{\alpha}}}}$
(as $X_i$ depends on the past $X^{i-1},Y^{i-1}$ only through
$Y_{i-1}$). Now, we are interested in expressing the conditional pmf
of the POST($\alpha$) channel recursively. This recursive formula
will be used later to find an input distribution that does not
utilize the feedback for the case of a POST Channel without feedback
and achieves the same output distribution, namely, a Markov chain
with transition probability $c{\bar \alpha
\alpha^{\frac{\alpha}{\bar{\alpha}}}}$.

Table \ref{t_P1_simple} presents the conditional pmf of a
POST($\alpha$) channel when $n=1.$
\begin{table}[h!]
\begin{center}
\begin{tabular}{|l||*{2}{c|}}\hline
\backslashbox{$Y_1$}{$X_1$}
&\makebox[3em]{0}&\makebox[3em]{1}\\\hline\hline 0
&1&$\alpha$\\\hline 1 &0& $\bar \alpha$\\ \hline \hline
\end{tabular}
\ \ \ \ \ \ \ \
\begin{tabular}{|l||*{2}{c|}}\hline
\backslashbox{$Y_1$}{$X_1$}
&\makebox[3em]{0}&\makebox[3em]{1}\\\hline\hline 0 &$\bar
\alpha$&0\\\hline 1 &$\alpha$& 1\\ \hline \hline
\end{tabular}
\caption{ Conditional probabilities $P(Y_1|X_1,s_0=0)$ (on Left) and $P(Y_1|X_1,s_0=1)$ (on the
Right) of the POST($\alpha$) channel. \label{t_P1_simple}}
\end{center}
\end{table}
Let us denote by $P_{n,0}$ and  $P_{n,1}$ the conditional matrices
of the channel given  the respective initial state, i.e., $s_0$ is 0
and 1, respectively. Namely,
\begin{eqnarray}\label{e_def_Pn}
P_{n,0} &\triangleq& P(y^n||x^n,s_0=0)\nonumber\\
P_{n,1} &\triangleq& P(y^n||x^n,s_0=1).
\end{eqnarray}
The columns of the matrices $P_{n,0}$ and  $P_{n,1}$  are indexed by
$x^n=(x_1,x_2,...,x_n)$ and the rows by  $y^n=(y_1,y_2,...,y_n)$
arranged via lexicographical order, where $x_1$ and $y_1$ are the
most significant bits and $x_n$ and $y_n$ are the least significant
bits. For instance, the conditional probabilities $P(y_1|x_1,s_0=0)$
and $P(y_1|x_1,s_0=1)$ are given in Table \ref{t_P1_simple}. Hence,
$P_{n,0}$ and $P_{n,1},$  for $n=1,$ are given by
\begin{equation}
P_{1,0} = \left[\begin{array}{cc} 1&\alpha\\
0& \bar \alpha  \end{array} \right] \ \ \ \ \ \ \ \
P_{1,1} =\left[\begin{array}{cc} \bar \alpha &0\\
\alpha & 1  \end{array} \right].
\end{equation}
\begin{table}[h!]
\begin{center}
\begin{tabular}{|l||*{4}{c|}}\hline
\backslashbox{$Y_1 Y_2$}{$X_1 X_2$}
&\makebox[3em]{00}&\makebox[3em]{01}&\makebox[3em]{10}&\makebox[3em]{11}\\\hline\hline
00
&1&$\alpha$  &$ \alpha$&$\alpha ^2$ \\\hline 01 &0& $\bar \alpha$ &0& $\alpha \bar \alpha$\\
\hline \hline 10 &0&$0$  &$\bar \alpha^2 $&$0$ \\\hline 11 &0& $0$ &$\bar \alpha \alpha $& $\bar \alpha$\\
\hline \hline
\end{tabular}
\caption{ Conditional probabilities $P(Y^2||X^2,s_0=0)$ .}
\label{t_P2_simple}
\end{center}
\end{table}
Table \ref{t_P2_simple} presents   $P(y^2||x^2,s_0=0)$ and the
corresponding matrix $P_{2,0} $ is given in eq.
(\ref{e_P_20_simple}).
\begin{equation}
P_{2,0} = \left[ \begin{array}{cccc} 1&\alpha  &\alpha&\alpha^2\\
0& \bar \alpha &0& \alpha \bar \alpha\\
0&0  &\bar \alpha ^2&0 \\
0& 0 &\bar \alpha \alpha& \alpha \end{array} \right]
\label{e_P_20_simple}
\end{equation}

From the channel definition, the following recursive relation holds
\begin{equation}
P_{n,0}=\left[ \begin{array}{cc}  1\cdot P_{n-1,0}& \alpha\cdot
P_{n-1,0}\\ 0\cdot P_{n-1,1} & \bar \alpha \cdot
P_{n-1,1}\end{array} \right]
\end{equation}
and
\begin{equation}
P_{n,1}=\left[ \begin{array}{cc}  \bar \alpha \cdot P_{n-1,0}& 0\cdot P_{n-1,0}\\
\alpha \cdot P_{n-1,1} & 1\cdot P_{n-1,1}\end{array} \right],
\end{equation}
where $P_{0,0}=P_{0,1}=1$ (i.e., the one by one unit matrix).

Recall that the  output process $\{Y_i\}_{i\geq 1}$ induced by the
input that achieves the feedback capacity, is a binary symmetric
Markov chain with transition probability $c{\bar \alpha
\alpha^{\frac{\alpha}{\bar{\alpha}}}}$, see (\ref{e_py}). Let
$p_0(y^n)$ and $p_1(y^n)$ denote the probability of $y^n$ given the
initial state 0 and 1, respectively. Hence, for $n=1$ we have
\begin{align}
p_0(y_1=0)&=p_1(y_1=1)=c \nonumber \\
p_0(y_1=1)&=p_1(y_1=0)=c{\bar \alpha
\alpha^{\frac{\alpha}{\bar{\alpha}}}}.\label{e_0_markov}
\end{align}
For $n\geq2,$
\begin{align}\label{e_markov}
 p_0(y_n| y^{n-1})=p_1(y_n| y^{n-1})=p_{y_{n-1}}(y_n).
\end{align}%
Now, we present this Markov process in a recursive way. Recall that
$P(y^n)$ is represented as  a (column) probability vector of
dimension $2^n$.
\begin{lemma}[vector representation of a pmf of a symmetric Markov process]
Let $Y^n$ be binary symmetric Markov with transition probability
$\delta$. Let $P_0(y^n)$ and $P_1(y^n)$ be the vector pmf when the
intial state is $0$ and $1$, respectively. One can describe the
vector pmf using the following recursive relation
\begin{equation}
P_{0}(y^n) = \left[\begin{array}{c} \bar \delta P_0(y^{n-1})\\
 \delta P_1(y^{n-1}) \end{array} \right] \ \ \ \ \text{ and} \ \ \ \ \
P_{1}(y^n) = \left[\begin{array}{c} \delta P_0(y^{n-1})\\
 \bar \delta P_1(y^{n-1}) \end{array} \right],\label{e_p_lemma_simple}
\end{equation}
where  $P_0(y^{0})=P_1(y^{0})=1,$ and $\left[\begin{array}{c}  u\\
 v \end{array} \right]$ denotes the column vector obtained by concatenating the two column vectors $u$ and $v$.
\end{lemma}

\begin{proof}
We prove this claim by induction. Note that, for $n=1$
(\ref{e_p_lemma_simple}) implies
\begin{align}
p_0(y_1=0)&=p_1(y_1=1)=\bar \delta \nonumber \\
p_0(y_1=1)&=p_1(y_1=0)=\delta.\label{e_0_markov_lemma}
\end{align}
Now, we need to show that, regardless of the initial state,
(\ref{e_markov}) holds. Assume $y_1=1$ and the initial state is 0,
then for any $y_2^n$
\begin{eqnarray}
p_{0}(y_n|y^{n-1})&=& \frac{p_{0}(y^{n})}{p_{0}(y^{n-1})}\nonumber \\
&=& \frac{p_{0}(1,y_2^{n})}{p_{0}(1,y_2^{n-1})}\nonumber \\
&=& \frac{\delta p_{1}(y_2^{n})}{\delta p_{1}(y_2^{n-1})}\nonumber \\
&=& \frac{p_{1}(y_2^{n})}{p_{1}(y_2^{n-1})}\nonumber \\
&=& p_{1}(y_n|y_2^{n-1}).\label{e_p_procedure}
\end{eqnarray}
In a similar way we obtain for any $y_2^n$ and any initial state
\begin{align}
p_0(y_n| y^{n-1})=p_1(y_n| y^{n-1})=p_{y_1}(y_n| y_2^{n-1}).
\end{align}
By repeating the procedure in (\ref{e_p_procedure}) $i$ times, we
obtain that for $1\leq i<n$ and $\forall y_2^n$
\begin{align}
p_0(y_n| y^{n-1})=p_1(y_n| y^{n-1})=p_{y_i}(y_n| y_{i+1}^{n-1}).
\end{align}
Now, note that choosing $i=n-1$ we obtain (\ref{e_markov}), which
means that the process is indeed Markov with transition probability
$\delta$.
\end{proof}

The following is our first main result:
\begin{theorem}\label{l_simple_y_channel}
Feedback does not increase the capacity of the POST($\alpha$) channel.
\end{theorem}

{\it Proof:} According to Corollary \ref{cor_pyn},  in order to show
that the nonfeedback capacity equals the feedback capacity, it
suffices to show that there exists an input pmf  $P(x^n)$ that
induces the optimal $P^*(y^n)$ of the feedback case, which is the
binary symmetric Markov chain with transition probability $c{\bar
\alpha \alpha^{\frac{\alpha}{\bar{\alpha}}}}$.

We will now find such a pmf by calculating
$P_1(x^n)=P_{n,1}^{-1}P_1(y^n)$ and $P_0(x^n)=P_{n,0}^{-1}P_0(y^n)$
and verifying that $P_0(x^n)$ and  $P_1(x^n)$ are indeed a valid
pmf. Recall that
\begin{equation}
 \left[\begin{array}{cc} A&B\\
0& D  \end{array} \right]^{-1}=\left[\begin{array}{cc} A^{-1}& -A^{-1}BD^{-1}\\
0& D^{-1}  \end{array} \right].
\end{equation}
Hence,

\begin{align}
P_{n,0}^{-1}&=\left[ \begin{array}{cc}  1\cdot P_{n-1,0}&
\alpha\cdot P_{n-1,0}\\ 0\cdot P_{n-1,1} & \bar \alpha \cdot
P_{n-1,1}\end{array} \right]^{-1} =\left[
\begin{array}{cc}  P_{n-1,0}^{-1}&  -\frac{\alpha}{\bar \alpha}P_{n-1,1}^{-1}\\
0 & \frac{1}{\bar \alpha} P_{n-1,1}^{-1}\end{array} \right]
\end{align}

\begin{align}
P_{n,1}^{-1}=\left[ \begin{array}{cc}  \bar \alpha \cdot P_{n-1,0}& 0\cdot P_{n-1,0}\\
\alpha \cdot P_{n-1,1} & 1\cdot P_{n-1,1}\end{array} \right]^{-1}=\left[ \begin{array}{cc}  \frac{1}{\bar \alpha}P_{n-1,0}^{-1}& 0\\
-\frac{\alpha}{\bar \alpha} P_{n-1,0}^{-1} &
P_{n-1,1}^{-1}\end{array} \right]
\end{align}


Now we compute $P_1(x^n)$ and $P_0(x^n).$

\begin{align}\label{e_PX0_recursive}
P_0(x^n)&=P_{n,0}^{-1}P_0(y^n)\nonumber\\
&= \left[
\begin{array}{cc}  P_{n-1,0}^{-1}&  -\frac{\alpha}{\bar \alpha}P_{n-1,1}^{-1}\\
0 & \frac{1}{\bar \alpha} P_{n-1,1}^{-1}\end{array} \right]
\left[\begin{array}{c} c P_0(y^{n-1})\\
 c{\bar
\alpha \alpha^{\frac{\alpha}{\bar{\alpha}}}}
P_1(y^{n-1}) \end{array} \right]\nonumber \\
&=  c\left[\begin{array}{c} P_0(x^{n-1})-{ \alpha^{\frac{1}{\bar{\alpha}}}} P_1(x^{n-1})\\
{ \alpha^{\frac{\alpha}{\bar{\alpha}}}}P_1(x^{n-1}) \end{array}
\right],
\end{align}

\begin{align}\label{e_PX1_recursive}
P_1(x^n)&=P_{n,1}^{-1}P_1(y^n)\nonumber\\
&= \left[ \begin{array}{cc}  \frac{1}{\bar \alpha}P_{n-1,0}^{-1}& 0\\
-\frac{\alpha}{\bar \alpha} P_{n-1,0}^{-1} &
P_{n-1,1}^{-1}\end{array} \right] \left[\begin{array}{c} c{\bar
\alpha \alpha^{\frac{\alpha}{\bar{\alpha}}}} P_0(y^{n-1})\\
 cP_1(y^{n-1}) \end{array} \right]\nonumber \\
&=  c\left[\begin{array}{c} {
\alpha^{\frac{\alpha}{\bar{\alpha}}}}P_0(x^{n-1}) \\
P_1(x^{n-1})-{ \alpha^{\frac{1}{\bar{\alpha}}}}P_0(x^{n-1})
 \end{array} \right],
 \end{align}
where $P_0(x^{0})=P_1(x^{0})=1$.

Now, we need to show that  the probability expressions are valid,
namely, nonnegative and sum to 1. The fact that they sum to 1 can be
seen from the recursion immediately by verifying that
$c(\alpha^{\frac{\alpha}{\bar{\alpha}}} +1 -{
\alpha^{\frac{1}{\bar{\alpha}}}})=1$.

In order to show the nonnegativity we need to show that
\begin{align}\label{e_what_need_to_prove}
P_0(x^{n-1})-\alpha^{\frac{1}{\bar{\alpha}}} P_1(x^{n-1})&\geq 0 \nonumber \\
P_1(x^{n-1})-\alpha^{\frac{1}{\bar{\alpha}}} P_0(x^{n-1})&\geq 0 .
\end{align}  For $n=1$ this is true since $\alpha^{\frac{1}{\bar{\alpha}}}\leq 1$ for $0\leq \alpha\leq 1$ (see Lemma \ref{l_in1_alpha} in the appendix).
The following lemma (Lemma \ref{l_nonegatovity_alpha}) states that
if (\ref{e_what_need_to_prove}) holds for $n-1$ it also holds for
$n$ and, by induction, we conclude that (\ref{e_what_need_to_prove})
holds for all $n$. \hfill \QED
\begin{lemma}\label{l_nonegatovity_alpha}
There exists a  $1\leq \beta\leq \alpha^{-\frac{1}{\bar{\alpha}}}$,
for which the condition
\begin{align}
\beta P_1(x^{n-1})&\geq P_0(x^{n-1}), \ \forall x^{n-1}, \nonumber \\
\beta P_0(x^{n-1})&\geq P_1(x^{n-1}),  \ \forall
x^{n-1},\label{e_lemma_asump_alpha}
\end{align}
implies
\begin{align}
\beta P_1(x^{n})&\geq P_0(x^{n}),  \ \forall x^{n}, \nonumber \\
\beta P_0(x^{n})&\geq P_1(x^{n}),  \ \forall x^{n}.
\label{e_lemma_to_proof_alpha}
\end{align}
\end{lemma}
\begin{proof}

Let's assume that (\ref{e_lemma_asump_alpha}) holds and we need to
show that (\ref{e_lemma_to_proof_alpha}) holds, which is equivalent
to showing the following four inequalities:
\begin{equation}\label{e1_alpha1}
\beta \left(c P_0(x^{n-1})-c{ \alpha^{\frac{1}{\bar{\alpha}}}}
P_1(x^{n-1})\right)\geq c{
\alpha^{\frac{\alpha}{\bar{\alpha}}}}P_0(x^{n-1}),
\end{equation}
\begin{equation}\label{e1_alpha2}
\beta c{ \alpha^{\frac{\alpha}{\bar{\alpha}}}}P_1(x^{n-1}) \geq c
P_1(x^{n-1})-c{ \alpha^{\frac{1}{\bar{\alpha}}}} P_0(x^{n-1}),
\end{equation}
\begin{equation}\label{e1_alpha3}
\beta c{ \alpha^{\frac{\alpha}{\bar{\alpha}}}}P_0(x^{n-1}) \geq c
P_0(x^{n-1})-c{ \alpha^{\frac{1}{\bar{\alpha}}}} P_1(x^{n-1}),
\end{equation}
\begin{equation}\label{e1_alpha4}
\beta \left(c P_1(x^{n-1})-c{ \alpha^{\frac{1}{\bar{\alpha}}}}
P_0(x^{n-1})\right)\geq c{
\alpha^{\frac{\alpha}{\bar{\alpha}}}}P_1(x^{n-1}).
\end{equation}
Because of the symmetry it suffices to show that (\ref{e1_alpha1})
and (\ref{e1_alpha2}) hold. We start by showing that
(\ref{e1_alpha1}) holds. The inequality (\ref{e1_alpha1}) is
equivalent to
\begin{equation}\label{e1_alpha}
\beta \left( P_0(x^{n-1})-{ \alpha^{\frac{1}{\bar{\alpha}}}}
P_1(x^{n-1})\right)\geq {
\alpha^{\frac{\alpha}{\bar{\alpha}}}}P_0(x^{n-1}),
\end{equation}
which can be further written as
\begin{equation}\label{e2_alpha}
 P_0(x^{n-1})(\beta-
\alpha^{\frac{\alpha}{\bar{\alpha}}})\geq  \beta
\alpha^{\frac{1}{\bar{\alpha}}} P_1(x^{n-1}),
\end{equation}
and as
\begin{equation}\label{e2_alpha}
 P_0(x^{n-1})\frac{\beta-
\alpha^{\frac{\alpha}{\bar{\alpha}}}}{ \beta
\alpha^{\frac{1}{\bar{\alpha}}}}\geq  P_1(x^{n-1}).
\end{equation}
This is true using the induction assumption in
(\ref{e_lemma_asump_alpha}) if
\begin{equation}
\frac{\beta- \alpha^{\frac{\alpha}{\bar{\alpha}}}}{\beta
\alpha^{\frac{1}{\bar{\alpha}}}}\geq \beta.
\end{equation}
Equivalently
\begin{equation}
 \alpha^{\frac{1}{\bar{\alpha}}} \beta^2  -\beta+\alpha^{\frac{\alpha}{\bar{\alpha}}}\leq 0.
\end{equation}
\begin{equation}\label{e_beta1}
\frac{1-\sqrt{1-4\alpha^{\frac{\alpha+1}{\bar{\alpha}}}}}{2\alpha^{\frac{1}{\bar{\alpha}}}}
\leq \beta \leq
\frac{1+\sqrt{1-4\alpha^{\frac{\alpha+1}{\bar{\alpha}}}}}{2\alpha^{\frac{1}{\bar{\alpha}}}},
\end{equation}
where the condition $4\alpha^{\frac{\alpha+1}{\bar{\alpha}}}\leq 1$
can be verified to hold for all $0\leq \alpha \leq 1,$ as shown in
Lemma \ref{l_inequality_4alpha} in the appendix.

Now let's  consider inequality (\ref{e1_alpha2}). We need to show
that
\begin{equation}\label{e1_alpha}
\beta { \alpha^{\frac{\alpha}{\bar{\alpha}}}}P_1(x^{n-1}) \geq
P_1(x^{n-1})-{ \alpha^{\frac{1}{\bar{\alpha}}}} P_0(x^{n-1}),
\end{equation}
or, equivalently,
\begin{equation}\label{e1_alpha}
 P_0(x^{n-1}) \geq  \frac{1-\beta {
\alpha^{\frac{\alpha}{\bar{\alpha}}}}}{\alpha^{\frac{1}{\bar{\alpha}}}}
P_1(x^{n-1}).
\end{equation}
This is true using the induction assumption in
(\ref{e_lemma_asump_alpha}) if
\begin{equation}
\frac{1}{\beta}\geq \frac{1-\beta {
\alpha^{\frac{\alpha}{\bar{\alpha}}}}}{\alpha^{\frac{1}{\bar{\alpha}}}},
\end{equation}
and equivalently
\begin{equation}
\alpha^{\frac{\alpha}{\bar{\alpha}}} \beta ^2 -\beta +
\alpha^{\frac{1}{\bar{\alpha}}}\geq 0.
\end{equation}
This holds if
\begin{equation}\label{e_beta2}
\beta\geq
\frac{1+\sqrt{1-4\alpha^{\frac{\alpha+1}{\bar{\alpha}}}}}{2\alpha^{\frac{\alpha}{\bar{\alpha}}}}.
\end{equation}
Combining (\ref{e_beta1}) with (\ref{e_beta2}) we obtain that there
exists a $\beta$ in the interval
\begin{equation}
\frac{1+\sqrt{1-4\alpha^{\frac{\alpha+1}{\bar{\alpha}}}}}{2\alpha^{\frac{\alpha}{\bar{\alpha}}}}
\leq \beta\leq
\frac{1+\sqrt{1-4\alpha^{\frac{\alpha+1}{\bar{\alpha}}}}}{2\alpha^{\frac{1}{\bar{\alpha}}}}
\end{equation}
that satisfies the lemma. Note that the interval is nonempty since
for $0\leq \alpha \leq 1$, $\alpha^{\frac{\alpha}{\bar{\alpha}}}\geq
\alpha^{\frac{1}{\bar{\alpha}}}.$ Finally, note that since
$\frac{1+\sqrt{1-4\alpha^{\frac{\alpha+1}{\bar{\alpha}}}}}{2}\leq 1$
we obtain that $\beta\leq  \alpha^{-\frac{1}{\bar{\alpha}}}.$
Furthermore, it is shown in Lemma \ref{l_inequlity_frac1} in the
appendix that
$\frac{1+\sqrt{1-4\alpha^{\frac{\alpha+1}{\bar{\alpha}}}}}{2\alpha^{\frac{\alpha}{\bar{\alpha}}}}\geq
1$ for any $0\leq \alpha\leq 1,$ and, therefore, also $\beta\geq 1.$
\end{proof}

In the proof that feedback does not increase the capacity of the
POST($\alpha$) channel (Theorem \ref{l_simple_y_channel}) we
actually found a recursive formula of the input distribution that
achieves the capacity (see
(\ref{e_PX0_recursive})-(\ref{e_PX1_recursive})). It is interesting
that the input distribution is not a Markov distribution and  has an
infinite memory but it induces a simple Markov distribution at the
output of the channel. In addition, as shown in the following lemma,
the input distribution is stationary.
\begin{lemma}[Stationarity of the input distribution]
The probabilities $P_0(x^n)$ and $P_1(x^n)$ that achieves the
capacity of the  POST($\alpha$) channel (i.e., achieve the maximum
in (11) when the mutual information is conditioned, respectively, on
the initial state being $0$ and $1$), and are specified in
(\ref{e_PX0_recursive})-(\ref{e_PX1_recursive}), are stationary,
namely, for any $n$ and $i<n$, and for any  $x^n$
\begin{equation}
p_0(x^i)=p_0(x_{n-i+1}^n),\ \  p_1(x^i)=p_1(x_{n-i+1}^n).
\end{equation}
\end{lemma}

\begin{proof}
We prove first the case $i=1$ and then generalize it to any $i$.
 For any $n$ and $i=1$ we have $p_0(x_{n})=\sum_{x^{n-1}}p_0(x^{n}).$
Note that applying this relation to (\ref{e_PX0_recursive}) we
obtain
\begin{equation}
 P_0(x_{n}) =  c\left[\begin{array}{c} 1-{ \alpha^{\frac{1}{\bar{\alpha}}}} \\
{ \alpha^{\frac{\alpha}{\bar{\alpha}}}} \end{array} \right],
\end{equation}
for all $n$. We can now generalize it to any $i$ by summing over
$x^{n-i}$ terms and we have
\begin{align}
P_0(x_{n-i+1}^n)
 &= c\left[\begin{array}{c}  P_0(x_{n-i+1}^{n-1})-{ \alpha^{\frac{1}{\bar{\alpha}}}}P_1(x_{n-i+1}^{n-1})\\
{ \alpha^{\frac{\alpha}{\bar{\alpha}}}} P_1(x_{n-i+1}^{n-1})
\end{array} \right].
\end{align}
Hence, this is the same iterative equation we have for $P_0(x^i)$ in
(\ref{e_PX0_recursive}), and since $P_0(x_1)=P_0(x_{n-i+1})$ we
obtain the same vector probability for all $n$. Similar proof holds
$P_1(x^n)$.
\end{proof}

\section{Binary POST$(a,b)$ channel }\label{s_post_ab}
In this section, we extend the scope and study  the capacity of what
we refer to as the  POST$(a,b)$ , which is a generalization of
POST($\alpha$) channel. The POST$(a,b)$ channel has two states and
in each state there is a binary channel with respective  parameters
$(a,b)$. We develop, using similar analyses as for the
POST($\alpha$) channel, simple conditions on $a,b$ for which,  if
satisfied,  feedback does not increase capacity. We then  show that
the
 conditions
are satisfied for all parameter values  $(a,b)\in [0,1]\times[0,1].$

\subsection{Definition of the POST$(a,b)$ channel}
Consider the POST channel depicted in  Fig. \ref{Ychannel_ab} with
the following behavior. When $y_{i-1}=0$, then the channel behaves
as a binary channel  with transition matrix
\begin{equation}
\left[ \begin{array}{cc} a & \bar b \\ \bar a & b \end{array} \right]
\end{equation}
 and when $y_{i-1}=1$ then it behaves as a binary channel with the transition matrix
\begin{equation}
\left[ \begin{array}{cc} b & \bar a \\ \bar b & a \end{array} \right].
\end{equation}
\begin{figure}[h!]{
\psfrag{d1}[][][1]{$y_{i-1}=0$} \psfrag{d2}[][][1]{$y_{i-1}=1$}  \psfrag{x}[][][1]{$x_{i}$}
\psfrag{y}[][][1]{$y_i$} \psfrag{a1}[][][0.9]{$a$} \psfrag{a1b}[][][0.9]{$\bar a\ \ \ $}
\psfrag{a2}[][][0.9]{$\bar b$}\psfrag{a2b}[][][0.9]{$ b$}
  \centerline{ \includegraphics[width=8cm]{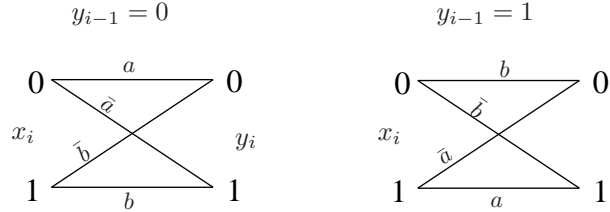}}
  \caption{POST$(a,b)$ channel.
  If $y_{i-1}=0$ then the channel behaves as DMC with parameters $(a,b)$ and if $y_{i-1}=1$ then the channel behaves as DMC with parameters $(b,a)$.}
 \label{Ychannel_ab}
}\end{figure} We refer to this channel as the POST$(a,b)$ channel.
POST($\alpha$) is a special case of POST$(a,b)$, where $a=1$ and
$b=\bar \alpha$.

Without loss of generality, we assume throughout that $a+b-1>0$. It
is easy to see that in the case where $a+b-1=0$ or, equivalently,
where $a=\bar b$,  the capacity is simply 0. Additionally, if
$a+b-1<0$ then $\bar a+\bar b>1;$ hence by relabeling the inputs $(0
\leftrightarrow 1)$ we obtain a new channel (with parameter $a',b'$
rather than $a,b$) where $a'=\bar a$ and $b'=\bar b$ and we have
$a'+b'-1>0.$

\subsection{Capacity of the POST$(a,b)$ channel with and without feedback}
Before considering the POST$(a,b)$ let us first consider the binary
DMC  with parameters $(a,b)$. The capacity of the binary DMC with
parameters $(a,b)$ was derived by Ash in \cite[Ex 3.7]{Ash65} by
applying \cite[Theorem 3.3.3]{Ash65} and is given by
\begin{equation}\label{e_c_dmc_ab}
C=\log\left[2^{\frac{\bar a H_b(b)-bH_b(a)}{a+b-1}}+ 2^{\frac{\bar b
H_b(a)-aH_b(b)}{a+b-1}} \right].
\end{equation}
 The  capacity achieving input distribution is
\begin{align}
P(x=0)&=c_0 \left(b2^{\frac{H(b)}{a+b-1}} -{\bar
b}2^{\frac{H(a)}{a+b-1}}\right),\nonumber
\\  P(x=1)&=c_0 \left(-\bar a2^{\frac{H(b)}{a+b-1}}
+{a}2^{\frac{H(a)}{a+b-1}}\right),
\end{align}
where $c_0$ is a normalizing coefficient so that the sum
$P(x=0)+P(x=1)$ is equal to  1. The induced output distribution is
\begin{align}
P(y=0)
 &= c_0(ab-\bar a\bar b)2^{\frac{H(b)}{a+b-1}}
\end{align}
\begin{align}
P(y=1)&= c_0(ab-\bar a\bar b)2^{\frac{H(a)}{a+b-1}}.
\end{align}

\begin{lemma}[Feedback capacity of POST$(a,b)$]
The feedback capacity of the POST($a,b$) channel  is the same as of
the memoryless DMC with parameters ($a,b$), which is given in
(\ref{e_c_dmc_ab}).
\end{lemma}
The proof follows the same arguments as the proof of the feedback
capacity of POST$(\alpha)$ in Lemma \ref{l_posta_feedback} and is,
therefore, omitted. The behavior  of the capacity as a function of
$(a,b)$ is depicted in Fig. \ref{f_post_ab}.

\begin{figure}[h!]{
\psfrag{a}[][][0.9]{$\ \ \ \ a$}
\psfrag{b}[][][0.9]{$b\ \ \ \ $} \psfrag{c}[][][0.8]{Capacity of
POST($a,b$)}
 \centerline{ \includegraphics[width=6cm]{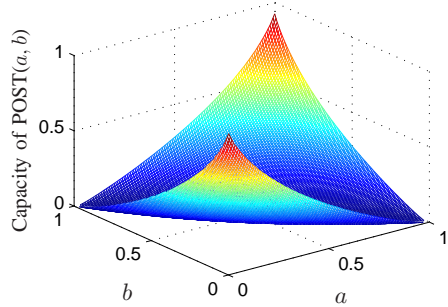}}
 \caption{The capacity of the POST($a,b$) channel with and without feedback. This is also the capacity of
 the binary DMC with parameters $(a,b)$\label{f_post_ab}}
}\end{figure}

We now present  sufficient conditions on $a,b$ implying that
feedback does not increase the capacity of the POST($a,b$) channel.
That these conditions are indeed sufficient we establish in the next
subsection. Define the following intervals:
\begin{align}
\mathcal L_1 &=\left\{ \max(\frac{\bar a}{\bar b}\gamma,
\frac{\gamma(\bar a+b)-\sqrt{\gamma^2(\bar a+b)^2-4a\bar b}}{2\bar
b})\leq  \beta \leq \frac{\gamma(\bar a+b)+\sqrt{\gamma^2(\bar
a+b)^2-4a\bar b}}{2\bar b}. \right\} \nonumber \\
\mathcal L_2 &=\left\{ \frac{( a+ \bar b)+\sqrt{( a+ \bar b)^2-4\bar
a b\gamma^2}}{2 b\gamma} \leq \beta\leq \frac{\bar a}{\bar b}\gamma\right\} \nonumber \\
\mathcal L_3 &=\left\{\beta\leq \min( \frac{\bar a}{\bar b}\gamma,
\frac{( a+ \bar
b)-\sqrt{( a+ \bar b)^2-4\bar a b\gamma^2}}{2 b\gamma})\right\} \nonumber \\
\mathcal L_4 &=\left\{ \beta \leq \min(\frac{b \gamma}{a},\frac{\gamma(\bar a+b)-\sqrt{\gamma^2(\bar a+b)^2-4a\bar b}}{2a}) \right\} \nonumber \\
\mathcal L_5 &=\left\{\frac{\gamma(\bar a+b)+\sqrt{\gamma^2(\bar a+b)^2-4a\bar b}}{2a}\leq \beta \leq \frac{b \gamma}{a} \right\} \nonumber \\
\mathcal L_6 &=\left\{\max( \frac{b \gamma}{a}, \frac{( a+\bar
b)-\sqrt{( a+\bar b)^2-4\bar a b\gamma^2}}{2\bar a\gamma}) \leq
\beta \leq \frac{( a+\bar b)+\sqrt{( a+\bar b)^2-4\bar a
b\gamma^2}}{2\bar a\gamma} \right\},\label{e_inter}
\end{align}
where $\gamma$ is defined as
\begin{equation}\gamma=2^{\frac{H(b)-H(a)}{a+b-1}}\label{e_gamma}.
\end{equation}
In addition, let
\begin{equation}\mathcal L_0 =\left\{ 1\le \beta\le
\min(\frac{a}{\bar{a}\gamma},\frac{b\gamma}{\bar{b}})\right\}
\end{equation}

\begin{lemma}[conditions to determine that feedback does not increase capacity of the POST$(a,b)$]\label{l_suffecienty_cond}
If the intersections of the intervals $\mathcal L_1\cup\mathcal
L_2\cup\mathcal L_3$ with $\mathcal L_4\cup\mathcal L_5\cup\mathcal
L_6$ and  $\mathcal L_0$  is nonempty then feedback does not
increase the capacity of the POST($a,b$) channel.
\end{lemma}

\begin{lemma}
The condition in  Lemma \ref{l_suffecienty_cond} holds for all POST
channel parameters $(a,b)$. Thus, feedback does not increase
capacity of POST($a,b$).
\end{lemma}

\begin{proof} It suffices to show that the either the interval
$(\mathcal L_1\cap\mathcal L_5\cap \mathcal L_0)$ or the interval
$(\mathcal L_2\cap L_6\cap \mathcal L_0)$ is non empty. First we
claim that the expression in the square roots of (\ref{e_inter}) are
nonnegative, i.e.,
\begin{align}
\gamma^2(\bar a+b)^2-4a\bar b& \geq 0\nonumber \\
( a+\bar b)^2-4\bar a b\gamma^2& \geq 0 \label{e_jiantao}
\end{align}
 as shown in Appendix \ref{app_sub_discrimnant} (Note that the second inequality follows from the first by switching between $a$ and $b$).

Recall that $a+b-1\geq 0$ which implies $a\geq \bar b$ and $\bar a
\leq b$. In addition,  assume that $\bar a a\leq b \bar b,$ and we
show in Appendix \ref{app_sub_a_geq_b} that together with  $a\geq
\bar b$ it implies that $a\geq b$. We now prove that if $\bar a
a\leq b \bar b,$ then $(\mathcal L_1\cap\mathcal L_5\cap \mathcal
L_0)$ is nonempty. (Similarly, one can show that if $\bar a a\geq b
\bar b,$ $(\mathcal L_2\cap L_6\cap \mathcal L_0)$ is non empty.) We
first want to claim that the intersection $(\mathcal L_1\cap\mathcal
L_5)$ is nonempty. Indeed, the lower bound of $\mathcal L_5$ is
smaller than the upper bound of $\mathcal L_1$ since
$\frac{1}{2a}\leq \frac{1}{2\bar b}$. The upper bound of $\mathcal
L_5$ is larger than the lower bound of $\mathcal L_1$ since
$\frac{b\gamma}{a}\geq\frac{\bar a\gamma}{\bar b}$ because $\bar a
a\leq b \bar b,$ and because,
\begin{eqnarray}
\frac{\gamma(\bar a+b)-\sqrt{\gamma^2(\bar a+b)^2-4a\bar b}}{2\bar
b}\leq\frac{b\gamma}{a},
\end{eqnarray}
where the  inequality follows the fact that the LHS is less or equal
to 1, as shown in Appendix \ref{app_sub_L1_L5_in}, and the RHS is
great or equal to 1, as show in Appendix \ref{app_sec_b_gamma}.

Now we need to show that the intersection  $(\mathcal
L_1\cap\mathcal L_5)$ with  $\mathcal L_0$ is still non empty. We
will first show that
$\min\{\frac{a}{\bar{a}\gamma},\frac{b\gamma}{\bar{b}}\}\geq
\frac{b\gamma}{a}$. Recall that $a\geq \bar b$ hence
$\frac{b\gamma}{\bar{b}}\geq \frac{b\gamma}{a}.$ In addition $\gamma
^2\leq \frac{a^2}{b\bar a}$ as shown in Appendix
\ref{app_sec_inq_gamma2}. Now we want to show that the lower bound
of $\mathcal L_0$ which is 1 is not larger than the upper bound of
the intersection $(\mathcal L_1\cap\mathcal L_5)$. First we claim
that the upper bound of $\mathcal L_5$ is larger than 1, i.e.,
$\frac{b\gamma}{a}\geq 1$ for $a\geq \bar b,$ as shown in Appendix
\ref{app_sec_b_gamma}. Finally, we need to show that  the upper
bound of $\mathcal L_1$ is larger than 1, i.e.,
\begin{equation}
\frac{\gamma(\bar a+b)}{2\bar b}\stackrel{}{\geq} 1,
\end{equation}
as shown in Appendix \ref{app_sec_bar_a}.
\end{proof}

%

\subsection{Deriving the sufficient conditions of  Lemma \ref{l_suffecienty_cond}}
{\it Proof of Lemma \ref{l_suffecienty_cond}:}
%
Let $P_{n,0}$ and $P_{n,1}$ be defined as in (\ref{e_def_Pn}).
Following the channel definition we have
\begin{equation}
P_{n,0}=\left[ \begin{array}{cc}  a\cdot P_{n-1,0}& \bar b \cdot P_{n-1,0}\\ \bar a\cdot
P_{n-1,1} & b \cdot P_{n-1,1}\end{array} \right]
\end{equation}
and
\begin{equation}
P_{n,1}=\left[ \begin{array}{cc}  b \cdot P_{n-1,0}& \bar a \cdot P_{n-1,0}\\
\bar b \cdot P_{n-1,1} & a \cdot P_{n-1,1}\end{array} \right]
\end{equation}
where $P_{0,0}=P_{0,1}=1$. Using the identity
\begin{equation}
\begin{bmatrix} \mathbf{A} & \mathbf{B} \\ \mathbf{C} & \mathbf{D} \end{bmatrix}^{-1} =
 \begin{bmatrix} \mathbf{A}^{-1}+\mathbf{A}^{-1}\mathbf{B}(\mathbf{D}-\mathbf{CA}^{-1}\mathbf{B})^{-1}\mathbf{CA}^{-1} & -\mathbf{A}^{-1}\mathbf{B}(\mathbf{D}-\mathbf{CA}^{-1}\mathbf{B})^{-1} \\ -(\mathbf{D}-\mathbf{CA}^{-1}\mathbf{B})^{-1}\mathbf{CA}^{-1} & (\mathbf{D}-\mathbf{CA}^{-1}\mathbf{B})^{-1} \end{bmatrix}
\end{equation}
%
%
%
%
%
%
%
we obtain
\begin{align}
P_{n,0}^{-1}&=\left[
\begin{array}{cc} \frac{b}{ba -\bar a \bar b}P_0^{-1}&  -\frac{\bar b}{ba -\bar a \bar b}P_1^{-1}\\
-\frac{\bar a}{ba -\bar a \bar b}P_0^{-1} & \frac{a}{ba -\bar a \bar a}P_1^{-1}\end{array}
\right]
\end{align}

\begin{align}
P_{n,1}^{-1}&=\left[
\begin{array}{cc} \frac{a}{ba -\bar a \bar b}P_0^{-1}&  -\frac{\bar a}{ba -\bar a \bar b}P_1^{-1}\\
-\frac{\bar b}{ba -\bar a \bar b}P_0^{-1} & \frac{b}{ba -\bar a \bar b}P_1^{-1}\end{array}
\right]
\end{align}

Now we compute $P_1(x^n)$ and $P_0(x^n)$
\begin{align}
P_0(x^n)&=P_{n,0}^{-1}P_0(y^n)\nonumber\\
&= \frac{1}{a+b-1}\left[
\begin{array}{cc} b P_0^{-1}&  -\bar bP_1^{-1}\\
-\bar aP_0^{-1} & aP_1^{-1}\end{array} \right]
 \left[\begin{array}{c} 2^{\frac{H(b)}{a+b-1}}P_0(y^{n-1})\\
 2^{\frac{H(a)}{a+b-1}}
P_1(y^{n-1}) \end{array} \right]\frac{1}{2^{\frac{H(b)}{a+b-1}}+2^{\frac{H(a)}{a+b-1}}} \nonumber \\
&=  \frac{1}{(a+b-1)(2^{\frac{H(b)}{a+b-1}}+2^{\frac{H(a)}{a+b-1}})} \left[\begin{array}{c}
b2^{\frac{H(b)}{a+b-1}}
P_0(x^{n-1})-\bar b 2^{\frac{H(a)}{a+b-1}} P_1(x^{n-1})\\
-\bar a2^{\frac{H(b)}{a+b-1}} P_0(x^{n-1})+ a 2^{\frac{H(a)}{a+b-1}} P_1(x^{n-1})
\end{array} \right],
\end{align}
\begin{align}
P_1(x^n)&=P_{n,1}^{-1}P_1(y^n)\nonumber\\
&=  \frac{1}{(a+b-1)(2^{\frac{H(b)}{a+b-1}}+2^{\frac{H(a)}{a+b-1}})} \left[\begin{array}{c}
a2^{\frac{H(a)}{a+b-1}}
P_0(x^{n-1})-\bar a 2^{\frac{H(b)}{a+b-1}} P_1(x^{n-1})\\
-\bar b2^{\frac{H(a)}{a+b-1}} P_0(x^{n-1})+ b 2^{\frac{H(b)}{a+b-1}} P_1(x^{n-1})
\end{array} \right],
 \end{align}
where $P_0(x^{0})=P_1(x^{0})=1$. We can rewrite $P_0(x^n)$ and
$P_1(x^n)$ follows:
\begin{align}
P_0(x^n) &=  \frac{1}{(a+b-1)(\gamma+1)} \left[\begin{array}{c} b\gamma
P_0(x^{n-1})-\bar b  P_1(x^{n-1})\\
-\bar a\gamma P_0(x^{n-1})+ a P_1(x^{n-1})
\end{array} \right],
\end{align}
\begin{align}
P_1(x^n)&=  \frac{1}{(a+b-1)(\gamma+1)} \left[\begin{array}{c} a
P_0(x^{n-1})-\bar a \gamma P_1(x^{n-1})\\
-\bar bP_0(x^{n-1})+ b \gamma P_1(x^{n-1})
\end{array} \right].
 \end{align}
We need to show that indeed the probability expressions are
valid, namely nonnegative and sum to 1. 
%
%
Showing the non-negativity of each of the terms in the above
expression is equivalent to showing $\forall n\ge 1$ and for all
$x^{n-1}$, \bea
\min\{\frac{a}{\bar{a}\gamma},\frac{b\gamma}{\bar{b}}\}P_0(x^{n-1})&\ge& P_1(x^{n-1})\nonumber\\
\min\{\frac{a}{\bar{a}\gamma},\frac{b\gamma}{\bar{b}}\}P_1(x^{n-1})&\ge&
P_0(x^{n-1}). \eea For $n=1$ this follows from the fact that
$\min\{\frac{a}{\bar{a}\gamma},\frac{b\gamma}{\bar{b}}\}\ge 1$ which
is proved in Appendix \ref{app_sec_gamma_and_d_non_neg}.
  To
prove for $n\geq 1$ we use the following lemma, whose proof appears
in Appendix \ref{app_prove_lemma_cond_beta}. \hfill \QED
\begin{lemma}\label{l_cond_beta}
If the condition in Lemma \ref{l_suffecienty_cond} holds then there
exists, $1\le \beta\le
\min\{\frac{a}{\bar{a}\gamma},\frac{b\gamma}{\bar{b}}\}$ such that
$\forall n$, the inequalities
\begin{align}
\beta P_1(x^{n-1})&\geq P_0(x^{n-1}), \ \forall x^{n-1}, \nonumber \\
\beta P_0(x^{n-1})&\geq P_1(x^{n-1}),  \ \forall
x^{n-1},\label{eq_a_b_n_1}
\end{align}
imply
\begin{align}
\beta P_1(x^{n})&\geq P_0(x^{n}),  \ \forall x^{n}, \nonumber \\
\beta P_0(x^{n})&\geq P_1(x^{n}),  \ \forall x^{n}. \label{eq_a_b_n}
\end{align}
\end{lemma}

\section{Does there exists a POST channel for which feedback increases the capacity?} \label{s_POST_increase_capacity}
For the nonfeedback capacity we do not have an analytical
expression, but only the infinite letter expression given in
(\ref{e_max_feedback}). In general, for any finite state channel we
have the following upper bounds \cite[Theorem 4.6.1]{Gallager68}
\cite[Theorem 15]{PermuterWeissmanGoldsmith09}
\begin{equation}\label{e_upper_bound}
C\leq \frac{1}{n} \max_{s_0}\max_{P(x^n)} I(X^n;Y^n|s_0)+\frac{\log|\mathcal S|}{n},
\end{equation}
for any integer $n\geq1$. The notation $|\mathcal S|$ refers to the
number of channel states. However, for the POST channel we have a
tighter upper bound given in the following corollary.
\begin{lemma}[Upper bound on the capacity of the POST channel]
For any POST channel  the following upper bound holds
\begin{equation}\label{e_upper_bound}
C\leq \frac{1}{n} \max_{s_0}\max_{P(x^n)} I(X^n;Y^n|s_0)
\end{equation}
for any integer $n\geq1.$
\end{lemma}
\begin{proof}
The proof of this upper bound follows from the subadditivity of the
sequence
$$\left\{\max_{s_0}\max_{P(x^N)} I(X^N;Y^N|s_0)\right\}_{N\geq 1}.$$
 In \cite[Theorem
4.6.1]{Gallager68} or in \cite[Theorem
16]{PermuterWeissmanGoldsmith09} it is proved that the sequence
$\{\max_{s_0}\max_{P(x^N)} I(X^N;Y^N|s_0)\}_{N\geq 1}+\log|\mathcal
S|$ is subadditive. We note that for the POST channel  the
conditioning on $S_n$ done in \cite[(4A.26)]{Gallager68} or in
\cite[(67)]{PermuterWeissmanGoldsmith09} is not needed since $S_n$
is a function of $Y_n$. Because of the subadditivity it follows that
the following limit exists and satisfies
\begin{equation}\label{e_limit_subadditive}
\lim_{N\to \infty} \max_{s_0}\max_{P(x^N)} I(X^N;Y^N|s_0)=\inf_N
\max_{s_0}\max_{P(x^N)} I(X^N;Y^N|s_0).
\end{equation}
From Fano's inequality it follows that the capacity is upper bounded by this limit and from
(\ref{e_limit_subadditive}) it follows that $\max_{s_0}\max_{P(x^N)} I(X^N;Y^N|s_0)$ upper bound
the capacity for any $N\geq 1$.
%
\end{proof}

%
%
%
Given the results that feedback does not increase the capacity of a
POST($a,b$), the question naturally arises: does there exist a
specific POST channel where feedback strictly increases the
capacity? The answer is affirmative. The idea is to find a POST
channel that consists of two states such that when there is
feedback, the optimal input distributions given the states differ
significantly among the different states. For the binary DMC it was
shown in \cite{Feder_uniform_input04} and independently in
\cite{GowthamKumar12_input_symbo} that the input probability that
achieves the capacities is in $[\frac{1}{e}, 1-\frac{1}{e}]$ for
each of the two alphabet symbols. Hence, in the case of a binary
post channel with feedback the optimal input probability as a
function of the state will not vary too much. But as we increase the
alphabet size we can construct a POST channel where the optimal
input probabilities as a function of the state would vary
significantly. Such a channel is presented in Fig. \ref{POST_many}.
\begin{figure}[h!]{
\psfrag{d1}[][][1]{$y_{i-1}=1,2,...,m$} \psfrag{d2}[][][1]{$y_{i-1}=m+1$}
\psfrag{x}[][][1]{$x_{i}$} \psfrag{y}[][][1]{$y_i$}\psfrag{n}[][][1]{$m$}\psfrag{n+1}[][][1]{$\
m+1$}
 \centerline{ \includegraphics[width=8cm]{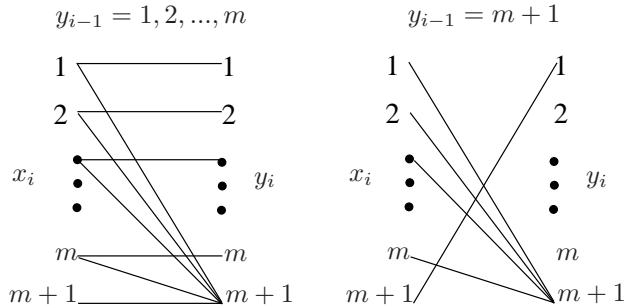}}
  \caption{ A POST channel where feedback increases capacity when $m$ is large. The probability associated with  each edge is either $\frac{1}{2}$ or 1.}
 \label{POST_many}
}\end{figure}

\begin{table}[h!]
\caption{Capacity with and without feedback of the POST channel
 in Fig. \ref{POST_many} as the alphabet grows. Evidently, feedback increases capacity for some values of $m$}\label{t_capacity_post}
\begin{center}

\begin{tabular}{||c|c|c|c||}
\hline \hline
 &  upper bound on capacity &  lower bound on feedback capacity & feedback capacity \\
$m$ & $\frac{1}{6}\max_{s_0 }\max_{P(x^6)} I(X^6;Y^6|s_0)$ &  $R=\frac{\log_2 m} { 3}$ & (\ref{e_feedback_capacity}) \\

\hline $2^0$ &  0.7918 &  0 & 0.7595\\
\hline $2^1$ & 0.8568 &  0.3333 & 0.8325\\
\hline $2^2$ & 0.9803 &  0.6667 & 1.0000\\
\hline $2^3$ &1.1711 &  1.0000 & 1.2599\\
\hline $2^4$ & 1.3865 &  1.3333 & 1.5366\\
\hline $2^5$ & 1.6098 &  1.6667 & 1.8260\\
\hline $2^6$ & 1.8374 &  2.0000 & 2.1252\\
\hline $2^7$ & 2.0683 &  2.3333 & 2.4319\\
\hline $2^8$ & 2.3019 &  2.6667 & 2.7444\\
\hline $2^9$ & 2.5376 &  3.0000 & 3.06140\\
\hline $2^{10}$ & 2.7751 &  3.3333 & 3.3818\\
\hline\hline
\end{tabular}

\end{center}
\end{table}

We can determine the feedback capacity analytically using the
following lemma, which is proved in Appendix
\ref{app_proof_l_feedback_cap_example}.
\begin{lemma}\label{l_feedback_cap_example}
The feedback capacity of the channel in Fig. \ref{POST_many}, is
given by,
\begin{eqnarray}\label{e_feedback_capacity}
C_{fb}=\max_{\gamma,\delta\in[0,1]}\Bigg{\{}\frac{2\delta}{2\delta+1+\gamma}\Big{(}\frac{\bar{\gamma}}{2}\log_2(m)+h_2(\frac{1+\gamma}{2})-(1-\gamma)
\Big{)}+\frac{1+\gamma}{2\delta+1+\gamma}\Big{(}h_2(\delta)\Big{)}\Bigg{\}}
\end{eqnarray}
\end{lemma}
\begin{corollary}
Note that as $m$ approaches infinity, $\frac{\bar{\gamma}}{2}\log_2(m)+h_2(\frac{1+\gamma}{2})-(1-\gamma)\simeq \frac{\bar{\gamma}}{2}\log_2(m)>>h_2(\delta)$, thus $C_{fb}\simeq\max_{\gamma\in[0,1]}\Bigg{\{}\frac{1-\gamma}{3+\gamma}\log_2(m)\Bigg{\}}=\frac{\log_2(m)}{3}$.
\end{corollary}

This gives us the intuition to suggest the following simple scheme
approximately capacity-achieving for large $m$. If $y_{i-1}\leq m$,
then transmit $\log_2(m)$ bits via input $X_i=1,2,...,m$. The
probability that these bits would be received at the decoder is
$\frac{1}{2}$. If $y_{i-1}= m+1$, then $X_i=m+1$. Thus, the rate
transmitted error free for this scheme is
\begin{equation}
R=\frac{\text{average bit transmitted}}{\text{average usage of channels}}=\frac{\frac{1}{2}
\log_2 m+\frac{1}{2} 0}{\frac{1}{2}1+\frac{1}{2}2}=\frac{\log_2m}{3}.
\end{equation}



We clearly see from Table \ref{t_capacity_post} that for $m\geq 2^2$
the feedback capacity is strictly larger than the non feedback
capacity. The difference increases with $m$.


\section{Conclusion and further research \label{s_conclusion}}\
We have introduced and studied the family of  POST channels and
showed, somewhat surprisingly,  that feedback does not increase the
capacity of the general $POST(a,b)$ channel. The proof is based on
finding the output probability that is induced by the input causal
conditioning pmf that optimizes the directed information when
feedback is allowed, and then proving that this output pmf can be
also be induced by an input distribution without feedback. There may
be a more direct way, that has thus far eluded us, for  proving that
feedback does not increase the capacity of the Simple POST channel.
We hope that the POST channel introduced in this paper will enhance
our understanding of capacity of finite state channels with and
without feedback, and help us to find simple capacity-achieving
codes.

\section*{Acknowledgement}
The authors are grateful to Jiantao Jiao who suggested the proof of
(\ref{e_jiantao}).



\appendices
\section{Concavity of directed information in $P(x^n||y^{n-1})$}
\label{app_concavity}
\begin{lemma}[Concavity of directed information in
$P(x^n||y^{n-1})$]\label{l_concavity_dir} Directed information
$I(X^n\to Y^n)$ is concave in $P(x^n||y^{n-1})$ for a fixed
$P(y^n||x^{n}).$
\end{lemma}
\begin{proof}
We need to show that for $0\leq \theta \leq 1$
\begin{align}&\sum_{x^n,y^n}(\theta p_1(x^n||y^{n-1})+\bar \theta p_2(x^n||y^{n-1}) p(y^n||x^n)\log
\frac{p(y^n||x^n)}{\sum_{x^n}(\theta p_1(x^n||y^{n-1})+\bar \theta
p_2(x^n||y^{n-1})) p(y^n||x^n)} \nonumber \\
&\geq \sum_{x^n,y^n}\theta p_1(x^n||y^{n-1})p(y^n||x^n)\log
\frac{p(y^n||x^n)}{\sum_{x^n} p_1(x^n||y^{n-1})
p(y^n||x^n)}+\bar\theta p_2(x^n||y^{n-1})p(y^n||x^n)\log
\frac{p(y^n||x^n)}{\sum_{x^n} p_2(x^n||y^{n-1}) p(y^n||x^n)}
\end{align}
This inequality may be written as
\begin{align}&\sum_{x^n,y^n}\theta p_1(x^n||y^{n-1}) p(y^n||x^n)\log
\frac{\sum_{x^n} p_1(x^n||y^{n-1}) p(y^n||x^n)}{\sum_{x^n}(\theta
p_1(x^n||y^{n-1})+\bar \theta p_2(x^n||y^{n-1})) p(y^n||x^n)}
\nonumber \\
& \ \ \ \ \ \ +\bar\theta p_2(x^n||y^{n-1}) p(y^n||x^n)\log
\frac{\sum_{x^n} p_2(x^n||y^{n-1}) p(y^n||x^n)}{\sum_{x^n}(\theta
p_1(x^n||y^{n-1})+\bar \theta p_2(x^n||y^{n-1})) p(y^n||x^n)}
 \geq 0
\end{align}
Furthermore,
\begin{align}
&\sum_{y^n}\theta  \left\{\sum_{x^n}p_1(x^n||y^{n-1}) p(y^n||x^n)\right\}\log
\frac{\sum_{x^n} p_1(x^n||y^{n-1}) p(y^n||x^n)}{\sum_{x^n}(\theta
p_1(x^n||y^{n-1})+\bar \theta p_2(x^n||y^{n-1})) p(y^n||x^n)}
\nonumber \\
& \ \ +\sum_{y^n}\bar\theta
  \left\{\sum_{x^n} p_2(x^n||y^{n-1})
p(y^n||x^n)\right\}\log \frac{\sum_{x^n} p_2(x^n||y^{n-1})
p(y^n||x^n)}{\sum_{x^n}(\theta p_1(x^n||y^{n-1})+\bar \theta
p_2(x^n||y^{n-1})) p(y^n||x^n)}
 \geq 0
\end{align}
Finally, note that the RHS is a sum of two divergences between pmf's
of $y^n$ and therefore it is positive.
\end{proof}

\section{Supporting inequalities}\label{app_inequalities}
\begin{lemma}\label{l_in1_alpha}
The  inequality \begin{equation}\alpha^{\frac{1}{\bar{\alpha}}}\leq
\label{e_in1_alpha} 1\end{equation} holds for $0\leq \alpha\leq 1$.
\end{lemma}
\begin{proof}
For $0\leq \alpha\leq 1$
\begin{equation}
\log \alpha \leq 0,
\end{equation}
which implies
\begin{equation}
\frac{1}{\bar{\alpha}} \log \alpha \leq 0,
\end{equation}
and equivalently
\begin{equation}
2^{\frac{1}{\bar{\alpha}} \log \alpha} \leq 2^0.
\end{equation}
Note that the last inequality is actually (\ref{e_in1_alpha}) for
nonnegative $\alpha$.
\end{proof}

\begin{lemma}\label{l_inequality_4alpha}
The  inequality
\begin{equation}4\alpha^{\frac{\alpha+1}{\bar{\alpha}}}\leq 1\end{equation} holds for
$0\leq \alpha\leq 1$.
\end{lemma}
\begin{proof}
By taking $\ln$ of both sides, we need to show
\begin{equation}\label{e_equi}
\frac{1+\alpha}{\bar \alpha}\ln \alpha +\ln 4\leq 0,
\end{equation}
which is equivalent to
\begin{equation}\label{e_up}
(1+\alpha)\ln \alpha+\bar \alpha\ln 4 \leq 0.
\end{equation}
In order to prove (\ref{e_up}) we claim that the RHS  increases in
$\alpha$ for $0\leq \alpha\leq 1$ and, therefore, the maximum value
is obtained at $\alpha=1$ and  is  0. In order to show that the RHS
of (\ref{e_up}) is increasing we need to show that its derivative is
nonnegative, i.e.,
\begin{equation}
\frac{1}{\alpha}+1+\ln \alpha-\ln 4\geq 0,
\end{equation}
or equivalently
\begin{equation}\label{e_econ}
1+\alpha+\alpha \ln \alpha-\alpha \ln 4\geq0.
\end{equation}
The RHS of (\ref{e_econ}) is a convex function and the minimum is
obtained when the derivative is zero, i.e.,
\begin{equation}
1+\ln\alpha+1-\ln4=0,
\end{equation}
which implies that $\ln \alpha=\ln \frac{4}{e^2}$. Hence, the
minimum value of the RHS of (\ref{e_econ}) is $1+
\frac{4}{e^2}+\frac{4}{e^2}\ln \frac{4}{e^2}- \frac{4}{e^2}\ln
4=1-\frac{4}{e^2},$ which is positive. Therefore (\ref{e_econ})
holds, which implies that (\ref{e_up}) holds, which implies that
(\ref{e_equi}) holds.
\end{proof}
\begin{lemma}\label{l_inequlity_frac1}
The following inequality holds
\begin{equation}\frac{1+\sqrt{1-4\alpha^{\frac{\alpha+1}{\bar{\alpha}}}}}{2\alpha^{\frac{\alpha}{\bar{\alpha}}}}\geq
1\end{equation}
 for any $0\leq \alpha\leq 1.$
\end{lemma}
\begin{proof}
Since $\alpha$ is nonnegative we need to show
\begin{equation}{\sqrt{1-4\alpha^{\frac{\alpha+1}
{\bar{\alpha}}}}}\geq{2\alpha^{\frac{\alpha}{\bar{\alpha}}}-1}.\end{equation}
This would be true if
\begin{equation}{{1-4\alpha^{\frac{\alpha+1}
{\bar{\alpha}}}}}\geq{4\alpha^{2\frac{\alpha}{\bar{\alpha}}}-4\alpha^{\frac{\alpha}{\bar{\alpha}}}+1},\end{equation}
which can be simplified to
\begin{equation}\label{e_in1}1\geq
\alpha^{\frac{\alpha}{\bar{\alpha}}}+\alpha^{\frac{1}
{\bar{\alpha}}},\end{equation}which can be written as
\begin{equation}\label{e_in1b}1\geq
\alpha^{\frac{\alpha}{\bar{\alpha}}}(1+\alpha).\end{equation} When
$\alpha\to 0^+$ we have equality, hence it suffices to show that the
equality holds after taking the derivative with respect to $\alpha$.
We use the equality $f'(\alpha)=f(\alpha)(\ln f(\alpha))'$ to find
\begin{equation}
(\alpha^{\frac{\alpha}{\bar{\alpha}}})'=\alpha^{\frac{\alpha}{\bar{\alpha}}}\left(\frac{\ln
\alpha}{\bar \alpha ^2}+\frac{1}{\bar \alpha}\right).
\end{equation}
Hence, applying the derivative on (\ref{e_in1b}) it remains to show
that
\begin{equation}
\alpha^{\frac{\alpha}{\bar{\alpha}}}\left(\frac{\ln \alpha}{\bar
\alpha ^2}+\frac{1}{\bar
\alpha}\right)(1+\alpha)+\alpha^{\frac{\alpha}{\bar{\alpha}}}\leq 0.
\end{equation}
or more simply
\begin{equation}
\ln \alpha+\bar \alpha+\frac{\bar{\alpha}^2}{1+\alpha}\leq 0.
\end{equation}
Note that if $\alpha=1$ there is equality. Hence it suffices to show
that the derivative of the LHS is non-negative for $0\leq \alpha\leq
1.$ I.e.,
\begin{equation}
\frac{1}{\alpha}-1-\frac{2\bar \alpha}{1+\alpha}-\frac{\bar
\alpha^2}{(1+\alpha)^2}\geq 0
\end{equation}
which is equivalent to
\begin{equation}
\frac{\bar \alpha}{\alpha}-\frac{\bar
\alpha(3+\alpha)}{(1+\alpha)^2}\geq 0
\end{equation}
and this is true if
\begin{equation}
(1+\alpha)^2-\alpha(3+\alpha)\geq 0\end{equation} which is
equivalent to
\begin{equation}
1-\alpha\geq 0,\end{equation} which is true.
\end{proof}

\section{Proof of Lemma \ref{l_cond_beta}
\label{app_prove_lemma_cond_beta}}
\begin{proof}
Suppose Eq. (\ref{eq_a_b_n_1}) holds; for Eq. (\ref{eq_a_b_n}) to
hold, we need to have $\forall x^{n-1}$, \bea
\beta (b\gamma P_0(x^{n-1})-\bar b  P_1(x^{n-1}))& \ge &a P_0(x^{n-1})-\bar a \gamma P_1(x^{n-1})\label{e1} \\
\beta (-\bar a\gamma P_0(x^{n-1})+ a P_1(x^{n-1}))& \ge &-\bar bP_0(x^{n-1})+ b \gamma P_1(x^{n-1})\label{e2} \\
\beta (a P_0(x^{n-1})-\bar a \gamma P_1(x^{n-1}))& \ge & b\gamma P_0(x^{n-1})-\bar b  P_1(x^{n-1})\label{e3} \\
\beta (-\bar bP_0(x^{n-1})+ b \gamma P_1(x^{n-1}))& \ge & -\bar
a\gamma P_0(x^{n-1})+ a P_1(x^{n-1}),\label{e4} \eea
 or equivalently, (\ref{e1}) and (\ref{e2}) become
  \bea
(b\gamma\beta-a)P_0(x^{n-1}) &\geq& (\bar b \beta-\bar a\gamma) P_1(x^{n-1}), \label{e1a} \\
(\bar b - \bar a \gamma \beta)  P_0(x^{n-1})&\geq& (b\gamma -a
\beta)P_1(x^{n-1}),\label{e2a} \eea
 and (\ref{e3}) and (\ref{e4}) become
 \bea
 (\bar b - \bar a \gamma \beta)  P_1(x^{n-1})&\geq& (b\gamma -a
\beta)P_0(x^{n-1})\label{e1b}\\
(b\gamma\beta-a)P_1(x^{n-1}) &\geq& (\bar b \beta-\bar a\gamma)
P_0(x^{n-1}). \label{e2b} \eea Because of the similarity of the
equations its enough to consider only (\ref{e1a}) and (\ref{e2a}).
Now we will consider a few cases.

{\bf The region of $\beta$ that satisfies (\ref{e1a}) and
(\ref{e2b}):} we will divide the treatment of (\ref{e1a}) (or
equivalently (\ref{e2b})) into two cases.

 {\bf Case 1:} $\beta \bar b-\bar a\gamma>0$ or equivalently
$\beta > \frac{\bar a}{\bar b}\gamma;$ Eq. (\ref{e1a}) becomes
\begin{equation}
P_0(x^{n-1})\frac{b\gamma\beta-a}{\bar b \beta-\bar a\gamma}\geq
P_1(x^{n-1}).
\end{equation}
By the assumption of the induction this would be true for all
$x^{n-1}$ if
\begin{equation}
\frac{b\gamma\beta-a}{\bar b \beta-\bar a\gamma}\geq \beta,
\end{equation}
or equvalently
\begin{equation}
\bar b \beta^2-\gamma(\bar a+b)\beta+a\leq 0.
\end{equation}
This implies
\begin{equation}
\frac{\gamma(\bar a+b)-\sqrt{\gamma^2(\bar a+b)^2-4a\bar b}}{2\bar
b}\leq  \beta \leq \frac{\gamma(\bar a+b)+\sqrt{\gamma^2(\bar
a+b)^2-4a\bar b}}{2\bar b},
\end{equation}
which is the interval $\mathcal L_1$ defined in (\ref{e_inter}).

{\bf Case 2:} $\beta \bar b-\bar a\gamma<0$ or equivalently $\beta <
\frac{\bar a}{\bar b}\gamma;$ Eq. (\ref{e1a}) becomes
\begin{equation}
P_0(x^{n-1})\frac{b\gamma\beta-a}{\bar b \beta-\bar a\gamma}\leq
P_1(x^{n-1}),
\end{equation}
which is true based on the induction assumption if
\begin{equation}
\frac{b\gamma\beta-a}{\bar b \beta-\bar a\gamma}\leq\frac{1}{\beta}.
\end{equation}
This is equivalent to
\begin{equation}
b\gamma \beta^2-(a+\bar b)\beta+\bar a\gamma\geq 0,
\end{equation}
and this is true if
\begin{equation}
\beta\geq \frac{( a+ \bar b)+\sqrt{( a+ \bar b)^2-4\bar a
b\gamma^2}}{2 b\gamma},
\end{equation}
which is the interval $\mathcal L_2,$  or
\begin{equation}
\beta\leq \frac{( a+ \bar b)-\sqrt{( a+ \bar b)^2-4\bar a
b\gamma^2}}{2 b\gamma},
\end{equation}
which is the interval $\mathcal L_3$.

{\bf The region of $\beta$ that satisfies (\ref{e2a}) and
(\ref{e1b}):} we will divide the treatment of (\ref{e2a}) (or
equivalently (\ref{e1b})) into two cases.

 {\bf Case 1:} $b \gamma -\beta a >0$ or equivalently
$\beta < \frac{b \gamma}{a};$ Eq. (\ref{e2a}) becomes
\begin{equation}
P_0(x^{n-1})\frac{\bar b-\bar a \gamma \beta}{b \gamma -\beta a
}\geq P_1(x^{n-1}).
\end{equation}
By the assumption of the induction this would be true for all
$x^{n-1}$ if
\begin{equation}
\frac{\bar b-\bar a \gamma \beta}{b \gamma -\beta a }\geq \beta,
\end{equation}
or equivalently
\begin{equation}
a\beta^2-\gamma( b+\bar a)\beta+\bar b\geq 0.
\end{equation}
This implies
\begin{equation}
\beta \leq \frac{\gamma(\bar a+b)-\sqrt{\gamma^2(\bar a+b)^2-4a\bar
b}}{2a},
\end{equation}
which is the interval $\mathcal L_4$, or
\begin{equation}
\beta \geq \frac{\gamma(\bar a+b)+\sqrt{\gamma^2(\bar a+b)^2-4a\bar
b}}{2a},
\end{equation}
which is the interval $\mathcal L_5.$

 {\bf Case 2:} $b \gamma -\beta a <0$ or equivalently
$\beta > \frac{b \gamma}{a};$ Eq. (\ref{e2a}) becomes
\begin{equation}
P_0(x^{n-1})\frac{\bar b-\bar a \gamma \beta}{b \gamma -\beta a
}\leq P_1(x^{n-1}).
\end{equation}
By the assumption of the induction this would be true for all
$x^{n-1}$ if
\begin{equation}
\frac{\bar b-\bar a \gamma \beta}{b \gamma -\beta a }\leq
\frac{1}{\beta}
\end{equation}
\begin{equation}
 \beta \bar b-\bar a \gamma \beta^2 \geq {b \gamma -\beta a }
\end{equation}
\begin{equation}
\bar a \gamma \beta^2 -\beta(a+\bar b) +b \gamma\leq 0.
\end{equation}
This implies
\begin{equation}
\frac{( a+\bar b)-\sqrt{( a+\bar b)^2-4\bar a b\gamma^2}}{2\bar
a\gamma} \leq \beta \leq \frac{( a+\bar b)+\sqrt{( a+ \bar
b)^2-4\bar a b\gamma^2}}{2\bar a\gamma}
\end{equation}
which is the interval $\mathcal L_6.$
\end{proof}

\section{Inequalities needed for the $POST(a,b)$ channel}
In this appendix we prove inequalities that are needed for proving
that feedback does not increase the capacity of $POST(a,b)$ channel.
All inequalities contains $\gamma$ which is defined in
(\ref{e_gamma}) and obviously $0\leq a\leq 1$ and $0\leq b\leq 1$.

\subsection{ $a\geq \bar b$ and $a\bar a\leq b \bar b$ implies that $a\geq b$. \label{app_sub_a_geq_b}}
\begin{proof}
Let $\rho_a=|\frac{1}{2}-a|$ and $\rho_b=|\frac{1}{2}-b|.$ Hence,
\begin{align}
a\bar a&=(\frac{1}{2}-\rho_a)(\frac{1}{2}+\rho_a)\nonumber \\
b\bar b&=(\frac{1}{2}-\rho_b)(\frac{1}{2}+\rho_b)
\end{align}
Since $a\bar a\leq b \bar b$ it follows that $\rho_a\geq \rho_b$.
And since $a\geq \bar b$ it follows that $a=\frac{1}{2}+\rho_a$ and
this implies $a\geq b$.
\end{proof}

\subsection{ $\gamma^2(\bar a+b)^2-4a\bar b\geq 0$. \label{app_sub_discrimnant}}
\begin{proof}
Assume first that $a\geq \bar b$
\begin{equation}
2^{2\frac{H(\bar b)-H(a)}{a-\bar b}}\geq \frac{4a \bar b}{(2-a-\bar
b)^2}
\end{equation}
Taking $\log$ on both sides we obtain
\begin{equation}\label{e_after_log}
2H(\bar b)-2H(a)\geq (a-\bar b)(\log 4\bar b+\log a-\log(2-a-\bar
b)^2)
\end{equation}
Note that if $a=\bar b$ we have equality, hence it suffices to prove
that the inequality holds after applying the derivative with respect
to $a$.
\begin{equation}\label{e_after_der1}
-2\ln \frac{\bar a}{a}\geq \ln 4\bar b+\ln a-2\ln(2-a-\bar
b)+(a-\bar b)(\frac{1}{a}+\frac{2}{2-a-\bar b})
\end{equation}
If $a=\bar b$ we note that there is equality hence it suffices to
prove that the inequality holds after taking the derivative with
respect to $a$.
\begin{equation}
\frac{2}{a}+\frac{2}{1-a}\geq \frac{1}{a}+\frac{2}{2-a-\bar
b}+\frac{1}{a}+\frac{2}{2-a-\bar b}+(a-\bar
b)(-\frac{1}{a^2}+\frac{2}{(2-a-\bar b)^2})
\end{equation}
After simplifying we obtain
\begin{equation}
\frac{a-\bar b}{a^2}+\frac{2}{\bar a}+\frac{2a+6\bar b-8}{(2-a-\bar
b)^2}\geq 0
\end{equation}
which after simple algebra it yields
\begin{equation}\label{e_after_der2}
\frac{a-\bar b}{a^2}+2\frac{(a-\bar b)(1-\bar b)}{\bar a(2-a-\bar
b)^2}\geq 0,
\end{equation}
and obviously the inequality holds since $a\geq \bar b$ and $b\leq
1$.

Now, if we consider $a\leq b$, then the sign $\geq$ in
(\ref{e_after_log}) should be replaces by the sign $\leq$. In
(\ref{e_after_der1} the sign should be the same, but in
(\ref{e_after_der2}) the sign should be the opposite again, which is
true.
\end{proof}

\subsection{${\gamma(\bar a+b)-\sqrt{\gamma^2(\bar a+b)^2-4a\bar b}}\leq 2\bar b,$ for $a\geq \bar b$ and $a\bar a\leq b \bar b$ \label{app_sub_L1_L5_in}}
\begin{proof}
We need to show
\begin{equation}
\gamma(\bar a+b)-2\bar b \leq \sqrt{\gamma^2(\bar a+b)^2-4a\bar b},
\end{equation}
since the RHS is nonnegative it suffices to show that
\begin{equation}
\gamma^2(\bar a+b)^2-4\bar b (\bar a+b)\gamma+4\bar b ^2\leq
{\gamma^2(\bar a+b)^2-4a\bar b},
\end{equation}
which simplifies to
\begin{equation}
\gamma(\bar a+b)\geq a+\bar b.
\end{equation}
After applying $\log$ on both sides, we need to show that
\begin{equation}
H(b)-H(a)\geq (a-\bar b)(\log (a+\bar b)-\log(\bar a+b)).
\end{equation}
If $a=\bar b$ we have equality, hence it suffices to show that after
we take derivative with respect to $a$ the inequality holds, i.e.,
\begin{equation}
\ln\frac{a}{\bar a}\geq \log\frac{a+\bar b}{\bar a+b}+ \frac{a-\bar
b}{a+\bar b}+\frac{a-\bar b}{\bar a+b}
\end{equation}
Again, if $b=\bar a$ we have equality, and we take the derivative
with respect to $b$, i.e.,
\begin{equation}
0\leq -\frac{1}{a+\bar b} -\frac{1}{\bar a+b}+\frac{2a}{(a+\bar
b)^2}+\frac{2\bar a}{(\bar a+b)^2}
\end{equation}
which after simplification equals to
\begin{equation}
0\leq \frac{a-\bar b}{(a+\bar b)^2}+\frac{\bar a-b}{(\bar a+b)^2}
\end{equation}
and since $a-\bar b=b-\bar a$, we only need to show that
\begin{equation}
 \frac{1}{(a+\bar b)^2}\leq\frac{1}{(\bar a+b)^2}
\end{equation}
which is true since $a+1-b\geq 1-a+b$ when $a\geq b$ and this
follows as shown in Appendix \ref{app_sub_a_geq_b}.
\end{proof}

\subsection{$\frac{b\gamma}{a}\geq 1$ for $a\geq \bar b$ \label{app_sec_b_gamma}}
\begin{proof}
We need to show that
\begin{equation}-b\log b -\bar b \log \bar b +a
\log a +\bar a\log \bar a \stackrel{}{\geq} (a-\bar b)(\log a  -
\log b)\end{equation} If $a=\bar b$, then both sides equal to zero.
Hence, its suffices to show that the inequality holds after applying
the derivative with respect to $a$. \begin{equation}\ln
\frac{a}{\bar a} \stackrel{}{\geq} (\ln a  - \log b)+1-\frac{\bar
b}{a}.\end{equation} We need to show that for $b\geq \bar a$
\begin{equation}\log \frac{b}{\bar a}+\frac{\bar b}{a}-1\geq
0\end{equation} If $b=\bar a$, then equality holds, hence its enough
to show that the derivative with respect to $b$ is positive. Namely,
\begin{equation}\frac{1}{b}-\frac{1}{a}\geq 0.\end{equation}
This is true since we also have $a\bar a\leq b \bar b$ and it
implies $a\geq b$ as shown in Appendix \ref{app_sub_a_geq_b}.
\end{proof}

%
%
\subsection{$\gamma ^2\leq \frac{a^2}{b\bar a}$ for $a\geq \bar b$ \label{app_sec_inq_gamma2}}
\begin{proof}
Taking $\log$ on both sides we need to show
\begin{equation}\frac{2h(b)-2h(a)}{a-\bar b}\stackrel{}{\leq} 2\log a-\log
b-\log \bar a.\end{equation}

%
%
%
After simple algebra we obtain
\begin{equation}(-b-\bar a) \log b -2\bar b \log \bar b+2\bar b \log a  +(\bar a+ b) \log \bar a\stackrel{}{\leq}0\end{equation}
We need to show that the inequality holds for $a\geq \bar b.$ Note
that when $a=\bar b.$ we obtain equality. Therefore its enough to
show that the derivative of RHS with respect to $a$ is negative and
therefore decreasing from $0$. At this point we transform all the
log to be natural base.
\begin{equation} \ln b +2\frac{\bar b}{a}  - \ln \bar a - 1-\frac{b}{\bar a} \stackrel{}{\leq}0\end{equation}
Again if  $b= \bar a$ we obtain equality. Now we take derivative
with respect to $b$ and need to show that its negative.
$$ \frac{1}{b} -\frac{2}{a}  -\frac{1}{\bar a} \stackrel{}{\leq}0$$
This is true since  $b\geq \bar a.$
\end{proof}

\subsection{$\frac{\gamma(\bar a+b)}{2\bar b}\stackrel{}{\geq} 1$ for $a\geq \bar b$ and $a\bar a\leq b\bar b$\label{app_sec_bar_a}}
\begin{proof}
We would like to show
\begin{equation}
\frac{\gamma(\bar a+b}{2\bar b}\stackrel{}{\geq} 1.
\end{equation}
Equivalently after taking $\log$ on both sides,
\begin{equation}
\frac{h(b)-h(a)}{a-\bar b}\stackrel{}{\geq} \log \frac{2\bar b}{\bar
a+b},
\end{equation}
\begin{equation}
-b\log b - \bar b\log \bar b+a\log a+\bar a\log \bar a
\stackrel{}{\geq} (a-\bar b)(\log2+\log \bar b-\log(\bar a+b)),
\end{equation}
Note that if $a=\bar b$ we have equality, hence it suffices to show
that for $b\geq \bar a$ the inequality holds after taking the
derivative with respect to $b$, i.e.,
\begin{equation}\label{e_in3}
\ln \frac{\bar b}{b}{\geq} \ln2+\ln \bar b-\ln(\bar a+b)+(a-\bar
b)(-\frac{1}{\bar b}-\frac{1}{\bar a+b})
\end{equation}
%
If $a=\bar b$ we get 0 on both sides.
Hence, it suffices to show that in we take derivative
with respect to $a$ (\ref{e_in3}) holds.
\begin{equation}
0\geq \frac{1}{\bar a+b}-\frac{1}{\bar b}-\frac{1}{\bar
a+b}-\frac{a-\bar b}{(\bar a+b)^2}
\end{equation}
and this trivially holds since the RHS is negative while the LHS is
0.
\end{proof}

\subsection{$\gamma \geq\frac{\bar b}{{b}}$ and $\gamma \leq
\frac{a}{\bar{a}}$ \label{app_sec_gamma_and_d_non_neg}}
\begin{proof}
Recall $ \gamma=2^{\frac{H(b)-H(a)}{a+b-1}} $. We prove only $\gamma
\geq\frac{\bar b}{{b}}$ and $\gamma \leq \frac{a}{\bar{a}}$ follows
from identical steps just replacing $a$ with $b$.

We need to show that
\begin{equation}
b2^{\frac{H(b)}{a+b-1}}\geq {\bar b}2^{\frac{H(a)}{a+b-1}}.
\end{equation}
Equivalently,
\begin{equation}
\log b +\frac{H(b)}{a+b-1}\geq \log{\bar b}+{\frac{H(a)}{a+b-1}}.
\end{equation}
\begin{equation}
(a+b-1)\log b -b\log b -\bar b \log \bar b \geq (a+b-1)\log{\bar
b}-a\log a-\bar a \log \bar a.
\end{equation}
\begin{equation}
-\bar a \log b   \geq a\log{\bar b}-a\log a-\bar a \log \bar a.
\end{equation}
\begin{equation}
\bar a \log \frac{\bar a}{b}+ a \log \frac{a}{\bar b}  \geq 0.
\end{equation}
Note that the last equality is a divergence expression between two
Bernoulli  distributions with parameters $a$ and $\bar b$ and hence
it's non negative.
\end{proof}

\section{Proof of Lemma
\ref{l_feedback_cap_example}}\label{app_proof_l_feedback_cap_example}
\begin{proof}
By symmetry, for any optimal distribution we will have,
$P(X_i=m+1|Y_{i-1}=k)$ equal for all $k\in\{1,\cdots,m\}$.  Hence
define \bea
 \gamma&\stackrel{\Delta}{=}&P(X_i=m+1|Y_{i-1}=k)\ \forall \ k\in\{1,\cdots,m\} \\
 \delta&\stackrel{\Delta}{=}&P(X_i=m+1|Y_{i-1}=m+1).
 \eea
 Also, by symmetry, $P(X_i=l|Y_{i-1}=k)=\frac{\bar{\gamma}}{m}\ \forall\  k,l\in\{1,\cdots,m\}$ and $P(X_i=l|Y_{i-1}=m+1)=\frac{\bar{\delta}}{m}\ \forall\ l\in\{1,\cdots,m\}$. Thus we have the following transition kernel for the induced output Markov Chain,
\begin{eqnarray}
P(Y_i=m+1|Y_{i-1}=m+1)&=&1-\delta\\
P(Y_i=m+1|Y_{i-1}=k)&=&\frac{1+\gamma}{2}\ \forall\ k\in\{1,2,\cdots,m\}\\
P(Y_i=1|Y_{i-1}=m+1)&=&\delta\\
P(Y_i=k|Y_{i-1}=m+1)&=&0\ \forall\ k\in\{2,\cdots,m\}\\
P(Y_i=k|Y_{i-1}=l)&=&\frac{\bar{\gamma}}{2m}\ \forall\
k,l\in\{1,2,\cdots,m\}.
\end{eqnarray}
Define the stationary distribution by
$\pi_k\stackrel{\Delta}{=}P(Y_i=k)$, for all $k\in\{1,\cdots,m+1\}$.
To obtain the stationary distribution, we have the following balance
equations, \bea
\sum_{k=1}^{m+1}\pi_k&=&1\\
\pi_{m+1}\delta&=&(\sum_{k=1}^{m}\pi_k)\frac{1+\gamma}{2}\\
\pi_1(1-\frac{\bar{\gamma}}{2m})&=&\pi_{m+1}\delta+(\sum_{k=2}^{m}\pi_k)\frac{\bar{\gamma}}{2m}\\
\pi_k(1-\frac{\bar{\gamma}}{2m})&=&(\sum_{l=1,l\neq
k}^{m}\pi_l)\frac{\bar{\gamma}}{2m}\ \forall \ k\in\{2,\cdots,m\}.
\eea We solve for the \textbf{Case: m=1}, first, for which the
equations are, \bea
\pi_1+\pi_2&=&1\\
\pi_{2}\delta&=&\pi_1\frac{1+\gamma}{2}, \eea which yields,
$\pi_1=\frac{\delta}{\delta+\frac{1+\gamma}{2}}$ and
$\pi_2=\frac{\frac{1+\gamma}{2}}{\delta+\frac{1+\gamma}{2}}$. Now
note, \bea
I(X_i;Y_i|Y_{i-1})&=&\pi_1 \Big{(}H(Y_i|Y_{i-1}=1)-H(Y_i|X_i,Y_{i-1}=1)\Big{)}+\pi_2 \Big{(}H(Y_i|Y_{i-1}=2)-H(Y_i|X_i,Y_{i-1}=2)\Big{)}\nonumber\\\\
&=&\pi_1\Big{(}h_2(\frac{1+\gamma}{2})-(1-\gamma)\Big{)}+\pi_2
h_2(\delta). \eea Now, as
$C_{fb}=\max_{\gamma,\delta\in[0,1]}I(X_i;Y_i|Y_{i-1})$, we obtain
\bea
C_{fb}&=&\max_{\delta,\gamma\in[0,1]}\Bigg{\{}\frac{\delta}{\delta+\frac{1+\gamma}{2}}\Big{(}h_2(\frac{1+\gamma}{2})-(1-\gamma) \Big{)}+\frac{\frac{1+\gamma}{2}}{\delta+\frac{1+\gamma}{2}}\Big{(}h_2(\delta) \Big{)}\Bigg{\}}\\
C_{fb}&=&\max_{\delta,\gamma\in[0,1]}\Bigg{\{}\frac{2\delta}{2\delta+1+\gamma}\Big{(}h_2(\frac{1+\gamma}{2})-(1-\gamma)
\Big{)}+\frac{1+\gamma}{2\delta+1+\gamma}\Big{(}h_2(\delta)
\Big{)}\Bigg{\}}.
\eea

Now, we deal with general \textbf{Case: $m\ge2$}, where from the
symmetry in balance equations, $\pi_k$ are equal $\forall \
k\in\{2,\cdots,m\}$, and hence we obtain the following on solving
the balance equations,
\begin{eqnarray}
\pi_1&=&A\frac{(m-1)\gamma+m+1}{1-\gamma}\\
\pi_k&=&A\ \forall\ k\in\{2,\cdots,m\}\\
\pi_{m+1}&=&A\frac{(1+\gamma)(m)}{\delta(1-\gamma)},
\end{eqnarray}
where constant $A=\frac{\delta(1-\gamma)}{2m\delta+m(1+\gamma)}$.
Note that
$I(X_i;Y_i|Y_{i-1}=k)=h_2(\underbrace{\frac{\bar{\gamma}}{2m},\cdots,\frac{\bar{\gamma}}{2m}}_{m\
\mbox{times}},\frac{1+\gamma}{2})-(1-\gamma)\ \forall\
k\in\{1,\cdots,m\}$. Also note that
$I(X_i;Y_i|Y_{i-1}=m+1)=h_2(\delta)$. Thus we obtain \bea
I(X_i;Y_i|Y_{i-1})&=&
(\sum_{k=1}^m\pi_m)\Big{(}h_2(\underbrace{\frac{\bar{\gamma}}{2m},\cdots,\frac{\bar{\gamma}}{2m}}_{m\  \mbox{times}},\frac{1+\gamma}{2})-(1-\gamma)\Big{)}+\pi_{m+1}h_2(\delta)\\
&=&\frac{2\delta}{2\delta+1+\gamma}\Big{(}h_2(\underbrace{\frac{\bar{\gamma}}{2m},\cdots,\frac{\bar{\gamma}}{2m}}_{m\
\mbox{times}},\frac{1+\gamma}{2})-(1-\gamma)\Big{)}+\frac{1+\gamma}{2\delta+1+\gamma}\Big{(}h_2(\delta)\Big{)}.
\eea Now, as
$C_{fb}=\max_{\gamma,\delta\in[0,1]}I(X_i;Y_i|Y_{i-1})$,
after some basic algebraic manipulation we obtain the following
expression:
\begin{equation}
C_{fb}=\max_{\delta,\gamma\in[0,1]}\Bigg{\{}\frac{2\delta}{2\delta+1+\gamma}\Big{(}\frac{\bar{\gamma}}{2}\log_2(m)+h_2(\frac{1+\gamma}{2})-(1-\gamma)
\Big{)}+\frac{1+\gamma}{2\delta+1+\gamma}\Big{(}h_2(\delta)\Big{)}\Bigg{\}}.
\end{equation}
\end{proof}

\end{document}